\documentclass{amsart}
\usepackage{hyperref}
\usepackage{graphicx}
\usepackage{curves}
\usepackage{enumerate}
\newcommand*{\mailto}[1]{\href{mailto:#1}{\nolinkurl{#1}}}

\newtheorem{theorem}{Theorem}[section]
\newtheorem{lemma}[theorem]{Lemma}

\newtheorem{proposition}[theorem]{Proposition}

\theoremstyle{definition}
\newtheorem{remark}[theorem]{Remark}

\newcommand{\R}{\mathbb{R}}
\newcommand{\Z}{\mathbb{Z}}

\newcommand{\C}{\mathbb{C}}
\newcommand{\T}{\mathbb{T}}

\newcommand{\nn}{\nonumber}
\newcommand{\be}{\begin{equation}}
\newcommand{\ee}{\end{equation}}
\newcommand{\bea}{\begin{eqnarray}}
\newcommand{\eea}{\end{eqnarray}}

\newcommand{\ol}{\overline}

\newcommand{\ti}{\tilde}

\newcommand{\id}{\mathbb{I}}
\newcommand{\I}{\mathrm{i}}
\newcommand{\E}{\mathrm{e}}

\newcommand{\clos}{\mathop{\mathrm{clos}}}
\newcommand{\re}{\mathop{\mathrm{Re}}}
\newcommand{\im}{\mathop{\mathrm{Im}}}

\DeclareMathOperator{\res}{Res}

\newcommand{\noprint}[1]{}

\newcommand{\si}{\sigma}
\newcommand{\la}{\lambda}

\numberwithin{equation}{section}

\begin{document}

\title[Rarefaction Waves for the Toda Equation]{Rarefaction Waves for the Toda Equation \\
via Nonlinear Steepest Descent}

\author[I. Egorova]{Iryna Egorova}
\address{B. Verkin Institute for Low Temperature Physics and Engineering\\ 47, Nauky ave\\ 61103 Kharkiv\\ Ukraine\\ and V.N. Karazin Kharkiv National University\\ 4, Svobody sq.\\ 61022 Kharkiv\\ Ukraine}
\email{\href{mailto:iraegorova@gmail.com}{iraegorova@gmail.com}}

\author[J. Michor]{Johanna Michor}
\address{Faculty of Mathematics\\ University of Vienna\\
Oskar-Morgenstern-Platz 1\\ 1090 Wien\\ Austria\\ and International Erwin Schr\"odinger
Institute for Mathematics and Physics\\ Boltzmanngasse 9\\ 1090 Wien\\ Austria}
\email{\href{mailto:Johanna.Michor@univie.ac.at}{Johanna.Michor@univie.ac.at}}
\urladdr{\href{http://www.mat.univie.ac.at/~jmichor/}{http://www.mat.univie.ac.at/\string~jmichor/}}

\author[G. Teschl]{Gerald Teschl}
\address{Faculty of Mathematics\\ University of Vienna\\
Oskar-Morgenstern-Platz 1\\ 1090 Wien\\ Austria\\ and International Erwin Schr\"odinger
Institute for Mathematics and Physics\\ Boltzmanngasse 9\\ 1090 Wien\\ Austria}
\email{\href{mailto:Gerald.Teschl@univie.ac.at}{Gerald.Teschl@univie.ac.at}}
\urladdr{\href{http://www.mat.univie.ac.at/~gerald/}{http://www.mat.univie.ac.at/\string~gerald/}}

\thanks{Discrete Contin. Dyn. Syst. {\bf 38}, 2007--2028 (2018)}
\thanks{Research supported by the Austrian Science Fund (FWF) under Grant No.\ V120.}

\begin{abstract}
We apply the method of nonlinear steepest descent to compute the long-time
asymptotics of the Toda lattice with steplike initial data corresponding to a rarefaction wave.
\end{abstract}

\maketitle

\section{Introduction}

In this paper we consider the doubly infinite Toda lattice
\begin{align} \label{tl}
	\begin{split}
\dot b(n,t) &= 2(a(n,t)^2 -a(n-1,t)^2),\\
\dot a(n,t) &= a(n,t) (b(n+1,t) -b(n,t)),
\end{split} \quad (n,t) \in \Z \times \R,
\end{align}
with steplike initial profile
\begin{align} \label{ini1}
\begin{split}	
& a(n,0)\to a, \quad b(n,0) \to b, \quad \mbox{as $n \to -\infty$}, \\
& a(n,0)\to \frac{1}{2} \quad b(n,0) \to 0, \quad \mbox{as $n \to +\infty$},
\end{split}
\end{align}
where $a>0$, $b\in\R$ satisfy the condition 
\be\label{main} 1 < b-2a.  
\ee 
This inequality implies that the spectra of the left and right background operators $H_\ell$ and $H_r$ 
have the following 
 mutual location: 
\[ 
\sup\si(H_r) < \inf\si(H_\ell).
\]
Here \[H_{\ell}y(n):= ay(n-1) + ay(n+1) + by(n),\ \ H_r y(n):=\frac{1}{2} y(n-1) +\frac 1 2 y(n+1),\ \ n\in\mathbb Z.\] In the case when $a=\frac{1}{2}$, the initial value problem \eqref{tl}--\eqref{main} is called rarefaction problem. We keep this name for an arbitrary $a>0$ and refer to the case $a=\frac{1}{2}$ as the classical rarefaction (CR) problem. 
The long-time asymptotics of the CR problem were studied rigorously by Deift et al.\ \cite{dkkz} in 1996 in the transitional region   where $\xi:=\frac{n}{t}\approx 0$ as $t\to +\infty$. 
To this end the authors applied the nonlinear steepest descent approach for vector Riemann--Hilbert (RH) problems. Using the same approach, our aim is to study
the region $\frac{n}{t}\in (-2 a +\varepsilon, -\varepsilon)\cup ( \varepsilon, 1-\varepsilon)$, where $\varepsilon>0$ is a sufficiently small number. 
Note that the regions $\frac{n}{t}\in (-\infty, -2 a-\varepsilon)$ and $\frac{n}{t}\in(1+\varepsilon, +\infty)$, which are called the soliton regions, can also be studied by the vector RH approach (see \cite{KTb} for decaying initial data $a=\frac{1}{2}$, $b=0$). 
Although the considerations for the soliton regions in the rarefaction case are more technical than in the decaying case, 
they are essentially the same and lead to a sum of solitons on the respective constant background. In our opinion, 
the classical inverse scattering transform with the analysis of the Marchenko equation 
provides this result easier (\cite{bdmek, BE, CK, toda}), and consequently, we will not study the soliton regions in this paper. Moreover, the transitional regions $\frac{n}{t}\approx 1$, $\frac{n}{t}\approx 0$, 
and $\frac{n}{t}\approx -2 a$ require further analysis and are also not the subject of the present paper.

For related results on the KdV equation using an ansatz based approach see \cite{LN}.
For results on the corresponding shock problem we refer to \cite{emt14, LN1, vdo} and the references therein.

In summary, we will show that there are four principal sectors with the following asymptotic behavior:
\begin{itemize}
\item In the region $n>t$, the solution $\{a(n,t), b(n,t)\}$ is asymptotically close to the constant right 
background solution $\{\frac{1}{2}, 0\}$ plus a sum of solitons corresponding to the eigenvalues $\la_j<-1$. 
\item In the region $0<n<t$, as $t\to\infty$ we have 
\be\label{ext} 
a(n,t)=\frac{n}{2t} +O\Big(\frac{1}{t}\Big),
\quad b(n,t)= 1+\frac{\frac{1}{2} -n}{t}
 +O\Big(\frac{1}{t}\Big).
\ee
\item In the region $-2a t<n< 0$, as $t\to\infty$ we have 
\be\label{ext2} 
a(n,t)=-\frac{n+1}{2t}+ O\Big(\frac{1}{t}\Big), \quad b(n,t)= b- 2a - \frac{n+\frac{3}{2}}{t}+ O\Big(\frac{1}{t}\Big).
\ee
\item In the region $n < -2a t$, the solution of \eqref{tl}--\eqref{main} is asymptotically close
to the left background solution $\{a, b\}$ plus a sum of solitons corresponding to the eigenvalues $\la_j>b+2a$.
\end{itemize}
The main terms of the asymptotics \eqref{ext} and \eqref{ext2} are  solutions of the Toda lattice equation. The terms $O(t^{-1})$ are 
uniformly bounded with respect to $n$  for $\varepsilon t \leq n\leq (1-\varepsilon)t$ in \eqref{ext} and for 
$(-2 a+\varepsilon)t \leq n\leq -\varepsilon t$ in \eqref{ext2}, where $\varepsilon>0$ is an arbitrary small value. 
Moreover, the  terms $O(t^{-1})$ are differentiable with respect to $t$, and the first derivatives 
are of order $O(\frac{n}{t^3})$. In the two middle regions we derive a precise formula for these error terms 
(see Theorem~\ref{theoras} and Proposition~ \ref{propos} below).

The following picture demonstrates the expected behavior of the Toda lattice solution in the middle regions. 
The numerically computed solution in Fig.~\ref{rarenool} corresponds to ``pure'' steplike initial data $a(n,0)=\frac{1}{2}$, $b(n,0)=0$ for $n\geq 0$ and $a(n,0)=0.4$, $b(n,0)=2$ for $n<0$. 
\begin{figure}[ht]
\centering
\includegraphics[width=5cm]{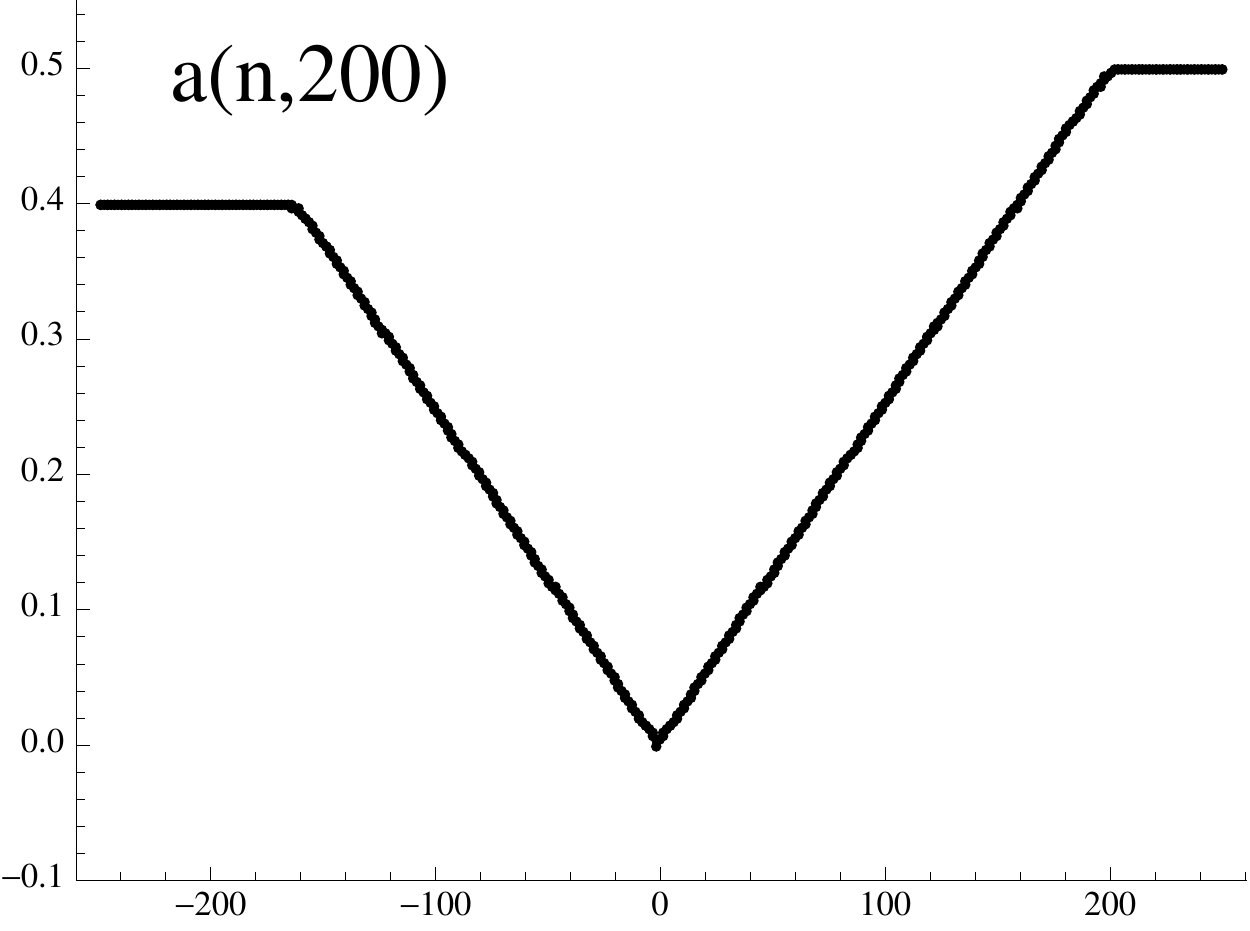}
\hspace{1cm}
\includegraphics[width=5cm]{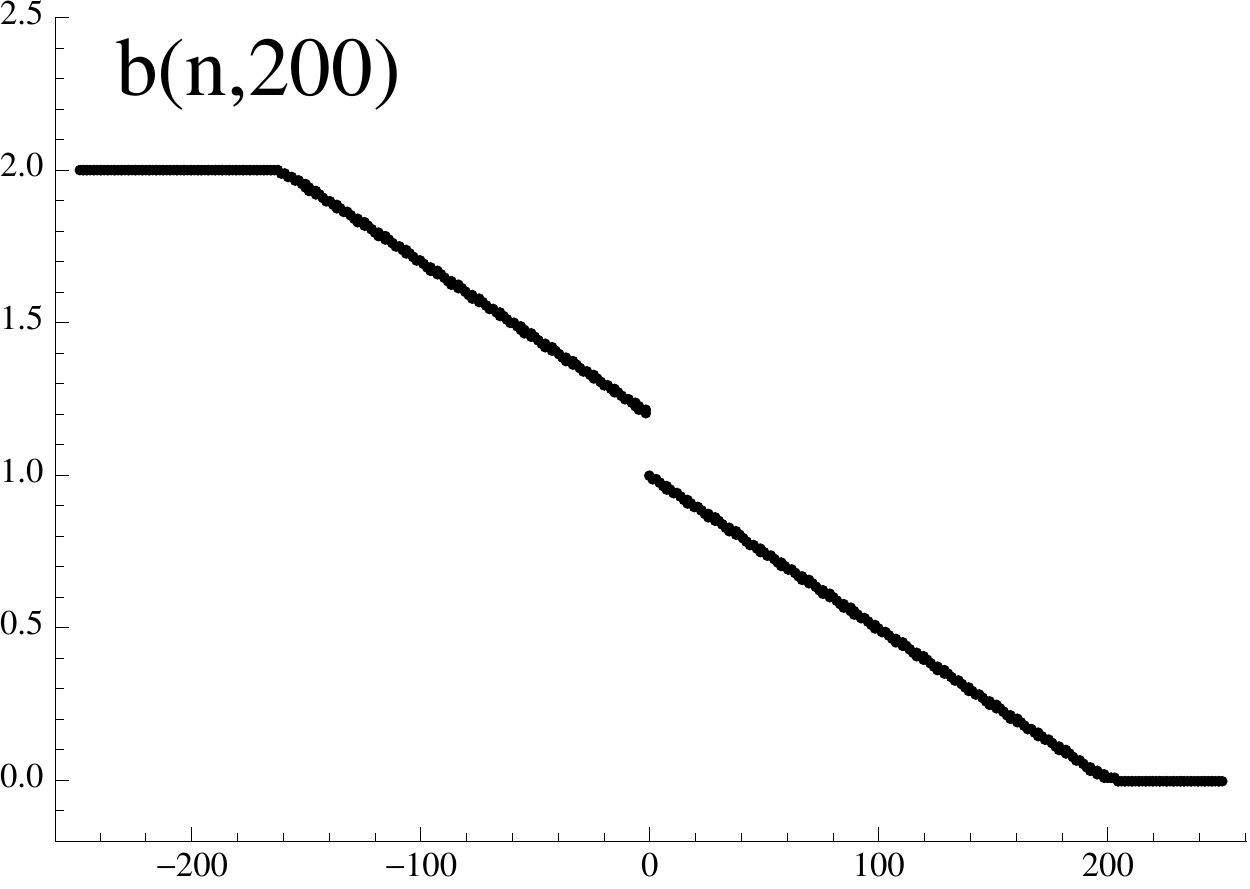}
\caption{ Toda rarefaction problem with non-overlapping background spectra $\sigma(H_{\ell})=[1.2, 2.8]$,   $\sigma(H_r)=[-1,1]$; $a=0.4$, $b=2$. } \label{rarenool} 
\end{figure}
The apparent line is due to the fact that neighboring points are very close due to the scaling. We observe that the analytically 
obtained asymptotics \eqref{ext}, \eqref{ext2} and the numerically computed asymptotics match well. 
In particular, in the analytic case the coefficient $b(n,t)$ has a jump in the transition region 
$\frac{n}{t}\approx 0$ as well.
An overview of the asymptotic solutions for \eqref{tl}--\eqref{ini1} with arbitrary constant $a>0$, $b\in\R$ can be found in \cite{m15}.

To simplify considerations we assume in addition to \eqref{ini1} that the initial data decay to their backgrounds exponentially fast
\be \label{decay}
 \sum_{n = 1}^{\infty} \E^{\nu n} \big( |a(-n,0) - a| + |b(-n,0)-b|
+ |a(n,0) - \tfrac{1}{2}| + |b(n,0)|\big) < \infty, 
\ee
where $\nu>0$ is an arbitrary small number.
This condition allows to continue the right reflection coefficient analytically to a small vicinity 
of the respective spectrum.

\section{Statement of the Riemann--Hilbert problem}

Let us first recall some elementary facts from scattering theory for Jacobi operators with steplike backgrounds 
from \cite{eg, emtstp2, emt3, emt5} (see also \cite[Chapter 10]{tjac} for general background).
The spectrum of the Jacobi operator $H(t)$ associated with the equation
\be\label{ht} 
H(t)y(n):=a(n-1,t)y(n-1) + b(n,t) y(n) + a(n,t) y(n+1)=\la y(n), \ \la\in\C,
\ee
consists of two intervals $[-1,1]\cup [b-2a,b+2a]$ of continuous spectrum 
with multiplicity one, plus a finite number of eigenvalues, 
$\{\lambda_j\}_{j=1}^{N} \subset \R \setminus ([-1,1]\cup [b-2a, b+2a])$. 
In addition to the spectral parameter $\la$ we will use two other  parameters $z$ and $\zeta$, connected with $\la$ by the Joukowsky transformation 
\be\label{spec5}
\la=\frac{1}{2}\left(z + z^{-1}\right)=b + a\left(\zeta + \zeta^{-1}\right),\quad |z|\leq 1, 
\quad |\zeta|\leq 1.
\ee
Introduce the Jost solutions $\psi(z,n,t)$, $\psi_\ell(z,n,t)$ of \eqref{ht}
with asymptotic behavior
\[
\lim_{n\to \infty} z^{-n}\psi(z,n,t) =1,\quad |z|\leq 1; \qquad
\lim_{n\to -\infty} \zeta^{n}\psi_\ell(z,n,t) =1, \quad |\zeta|\leq 1.
\] 
Denote
\[ %\label{nott}
q_1=z(b-2a),\quad q_2=z(b+2a), \quad z_j=z(\la_j), \quad j=1,\dots,N,
\]
where
\[
z(\la)=\la-\sqrt{\la^2-1}.
\]  
The points $z=-1$ and $z=1$ correspond to the edges of the spectrum of the right background Jacobi operator, and $q_2$ and $q_1$ correspond to the respective edges of the left background operator. We will call the points 
$z_j$ discrete spectrum.
Denote $\mathbb T=\{z: |z|=1\}$ and $\mathbb D=\{z: |z|<1\}$. 
The map $z\mapsto \la$ is one-to-one between the closed domains $\mathfrak D:=\clos(\mathbb D\setminus[q_2, q_1])$ and $\mathcal D:=\clos(\C\setminus([-1,1]\cup[b-2a,b+2a]))$.
We treat closure as adding to the boundary the points of the upper and lower sides along the cuts, while considering them as distinct points. Since the function $\zeta^n \psi_\ell(z)$ is in fact an analytical function of $\zeta$, 
it takes complex conjugated values on the sides of the cut along the interval $I:=[q_2, q_1]$, which we denote as $I\pm\I0$. Note that $z\in I-\I 0$ corresponds to the arc $\zeta\in \{|\zeta|=1, \im \zeta< 0\}$. The Jost solution $\psi_\ell(z)$  takes equal real values at $z, z^{-1}\in\mathbb T$, which yields the respective properties of the scattering matrix 
as a function of $z$.  This matrix consist of the right (resp., left) reflection coefficient $R(z,t)$ (resp. $R_\ell(z,t)$), defined for $|z|=1$ (resp., $|\zeta|=1$), and the transmission coefficients $T(z,t)$ and $T_\ell(z,t)$ defined on $\mathfrak D$. They are connected by the scattering relations
\[%\label{pst}
\aligned T(z,t)\psi_\ell(z,n,t)&=\ol {\psi(z,n,t)} + R(z,t)\psi(z,n,t), \quad |z|=1,\\
 T_\ell(z,t)\psi(z,n,t)&=\ol {\psi_\ell(z,n,t)} + R_\ell(z,t)\psi_\ell(z,n,t),\quad z\in I\pm\I 0.
\endaligned
\]
Moreover, let $z_j$, $1\le j \le N$, (note $|z_j|<1$) be the eigenvalues and set
\begin{align}\nn%\label{defphi}
\gamma_j(t)&:=\Big(\sum_{\Z}\psi^2(z_j,n,t)\Big)^{-2},\quad j=1,\dots,N, \\ \label{defchi}
\chi(z,t)&:=-\lim_{\varepsilon \to 0} \ol {T(z-\I\varepsilon, t)} T_\ell (z-\I\varepsilon, t), \quad z\in I.
\end{align}
The set of right scattering data
\be\label{scatt}
\{R(z,t), z\in\mathbb T; \ \chi(z,t), z\in I; \ (z_j, \gamma_j(t)), 1 \leq j \leq N \}
\ee
defines the solution of the Toda lattice uniquely. Under condition 
\eqref{decay}, it has the following properties (we list only those relevant for the present paper, see \cite{dkkz, emtstp2}):
\begin{itemize}
 \item  The function $R(z,t)$ is continuous and $R(z^{-1},t)=\ol{R(z,t)}=R^{-1}(z,t)$ for $z\in\T$.
We have $R(-1)=-1$ if $z=-1$ is non-resonant\footnote{The point $\hat z\in\{-1,1,q_1,q_2\}$ is called a resonant point if $W(\hat z,t)=0$, where $W(z,t):=a(n-1,t)(\psi_\ell(z,n-1,t)\psi(z,n,t) - \psi_\ell(z,n,t)\psi(z,n-1,t))$ is the Wronskian of the Jost solutions. 
If $W(\hat z, t)\neq 0$, then $\hat z$ is non-resonant.}
and $R(-1)=1$ if $z=-1$ is resonant. The function $R(z)$  can be continued analytically in the annulus $\E^{-\nu}<|z|<1$.
\item The right transmission coefficient $T(z,t)$ can be restored uniquely from \eqref{scatt} for $z\in\mathfrak D$. It is a meromorphic function with simple poles at $z_j$.
\item The function $\chi(z,t)$ is continuous for $z\in (q_2,q_1)$ and vanishes at $q_i$, $i=1,2$, iff 
$q_i$ is a non-resonant point. If $q_i$ is a resonant point, $\chi(z)=(z-q_i)^{-1/2}$. The transmission coefficient has the same behavior at $q_i$ as $\chi(z,t)$.
%\item The time evolution of the scattering data is given by
%\begin{align}\label{chin}
%\chi(z,t)& =\chi(z)\E^{(z-z^{-1})t}, \quad z\in I,\\ \label{rr}	
%R(z,t)&=R(z)\E^{(z-z^{-1})t}, \quad z\in \T,\\ \label{gamj}
%\gamma_j(t)&=\gamma_j\E^{(z_j-z_j^{-1})t}, \qquad j=1,...,N,
%\end{align}
%where $\chi(z)=\chi(z,0),\ R(z)=R(z,0),\ \gamma_j=\gamma_j(0).$
\end{itemize}

On $\mathfrak D$ we define a vector-valued function $m(z)=(m_1(z,n,t), m_2(z,n,t))$, 
\be \label{defm}
m(z,n,t) =
\begin{pmatrix} T(z,t) \psi_{\ell}(z,n,t) z^n,  & \psi(z,n,t)  z^{-n} \end{pmatrix}. 
\ee
\begin{lemma}[\cite{emt14}]\label{asypm}
The components of $m(z,n,t)$ have the following asymptotic behavior
as $z \to 0$ 
\be \label{asm1}
\aligned
m_1(z,n,t) & = \prod_{j=n}^\infty 2a(j,t) \Big(1 + 2 z \sum_{m=n}^\infty b(m,t)\Big) + O(z^2),\\ 
m_2(z,n,t)&=\prod_{j=n}^\infty (2a(j,t))^{-1}\Big(1 - 2z \sum_{m=n+1}^\infty b(m,t)\Big)\Big)
+  O(z^2).
\endaligned
\ee
\end{lemma}

Evidently, $m_1(z)$ is a meromorphic function with poles at $z_j$.
Let us extend $m$ to $\{z:|z|>1\}\setminus I^*$, $I^*:=[q_2^{-1}, q_1^{-1}]$, 
by  $m(z^{-1}) = m(z) \si_1$, 
where $\si_1$ is the first Pauli matrix. Recall that the Pauli matrices are given by 
\[%\label{pauli}
\si_1=\begin{pmatrix} 0  & 1 \\
1 & 0 \end{pmatrix}, \quad
\si_2= \begin{pmatrix} 0  & - \I \\
\I & 0 \end{pmatrix}, \quad
\si_3= \begin{pmatrix} 1  & 0 \\
0 & -1 \end{pmatrix},
\]
we will also use them in abbreviations as for example
\[% \label{matsi}
[d(z)]^{- \si_3} := \begin{pmatrix} d^{-1}(z) & 0 \\
0 & d(z) \end{pmatrix}.
\]
Moreover, $m_2(z)$ is a meromorphic function in $\{z: |z|>1\}\setminus I^*$ 
with poles at $z_j^{-1}$.

By definition, the vector function $m(z)$, $z\in \C$, has jumps along the unit circle and along the intervals $I$ and $I^*$.
The statement of the respective Riemann-Hilbert problem with pole conditions is given in \cite{dkkz}. We will not formulate this problem here, but instead give an equivalent statement which is valid in the domain $\xi:=\frac{n}{t}\in (-2a,0)\cup (0,1)$ we are interested in. In this domain we can reformulate the initial meromorphic RH problem as a holomorphic RH problem
as in \cite{dkkz, KTb}. We skip the details and only provide a brief outline below.

Throughout this paper, $m_+(z)$ (resp.\ $m_-(z)$) will denote the limit
of $m(p)$ as $p\to z$ from the positive (resp.\ negative) side of an oriented contour $\Sigma$. Here
the positive (resp.\ negative) side is the one which lies to the left (resp.\ right) as one 
traverses the contour in the direction of its orientation. Using this notation implicitly 
assumes that these limits exist in the sense that $m(z)$ extends to a continuous function 
on the boundary. Moreover, all contours are symmetric with respect to the map $z\mapsto z^{-1}$, i.e., they contain with each point $z$ also $z^{-1}$.
The orientation on these contours should be chosen in such a way that the following symmetry is preserved for the jump matrix of the vector RH problem and for its solution.

\vskip 2mm

{\bf Symmetry condition}. {\it Let $\hat\Sigma$ be a symmetric oriented contour. Then the jump matrix of the vector problem $m_+(z)=m_-(z)v(z)$ satisfies 
\be \label{matsym}
(v(z))^{-1}=\sigma_1 v(z^{-1}) \sigma_1,\quad z\in\hat\Sigma.
\ee
Moreover,} 
\be\label{symto}
m(z)=m(z^{-1})\sigma_1, \quad z\in\C\setminus\hat\Sigma.
\ee
Most of our transformations are conjugations with diagonal matrices, so it is convenient to use the following
\begin{lemma}[Conjugation, \cite{KTa}] \label{lem:conjug}
Let $m$ be the solution on $\C$ of the RH problem  $m_+(z)=m_-(z) v(z)$, $z \in \hat\Sigma$,  
which satisfies the symmetry  condition as above.
 Let $d: \C\setminus  \Sigma\to\C$ be a sectionally analytic function. Set
\be\label{mD}
\ti{m}(z) = m(z) [d(z)]^{-\sigma_3},
\ee
then the jump matrix of the problem $\ti{m}_+=\ti{m}_- \ti{v}$ is given by
\[
\ti{v} =
\left\{
\begin{array}{ll}
  \begin{pmatrix} v_{11} & v_{12} d^{2} \\ v_{21} d^{-2}  & v_{22} \end{pmatrix}, &
  \quad p \in \hat\Sigma \setminus  \Sigma, \\[4mm]
  \begin{pmatrix} \frac{d_-}{d_+} v_{11} & v_{12} d_+ d_- \\
  v_{21} d_+^{-1} d_-^{-1}  & \frac{d_+}{d_-} v_{22} \end{pmatrix}, &
  \quad p \in \Sigma.
\end{array}\right.
\]
If $d$ satisfies  $d(z^{-1}) = d(z)^{-1}$ for $z \in \C\setminus \Sigma$,
then the transformation \eqref{mD} respects the symmetry condition.
\end{lemma}
Recall that the behavior of the solution of the RH problem is determined mostly by the behavior of the phase function
\be  \label{Phi}
\Phi(z) = \Phi(z,\xi)=\frac{1}{2} \big(z - z^{-1}\big) + \xi\log z.
\ee
Let $\xi\in (0,1)$. Part of the eigenvalues lie in the domain $\re\Phi(z)>0$ (namely $z_j\in (-1,0)$), while 
the remaining eigenvalues belong to the set $\re\Phi(z)<0$. The pole conditions at the eigenvalues are 
given by (\cite{KTb})
\be \nn%\label{eq:polecond}
\aligned
\res_{z_j} m(z) &= \lim_{z\to z_j} m(z)
\begin{pmatrix} 0 & 0\\ - z_j \gamma_j \E^{2 t\Phi(z_j)}  & 0 \end{pmatrix},\\
\res_{z_j^{-1}} m(z) &= \lim_{z\to z_j^{-1}} m(z)
\begin{pmatrix} 0 & z_j^{-1} \gamma_j \E^{2 t\Phi(z_j)} \\ 0 & 0 \end{pmatrix}.
\endaligned
\ee
The pole conditions can be replaced by jump conditions on small curves around the eigenvalues
as in \cite{KTa}. Let 
\be\label{Blaschke}
P(z)= \prod_{z_j \in (-1,0)} |z_j| \frac{z - z_j^{-1}}{z-z_j}
\ee
be the Blaschke product corresponding to $z_j$ in the domain $\re\Phi(z)>0$  
(if any). Note that $P(z)$ satisfies $P(z^{-1})=P^{-1}(z)$. Let $\delta$ be sufficiently small 
such that the circles $\mathbb T_j=\{ z : |z - z_j|=\delta\}$ around the eigenvalues 
do not intersect and lie away from $\T \cup I$ (the precise value of $\delta$ will be chosen later).
Set 
\[
A(z)= \begin{cases} 
\begin{pmatrix}1& \frac{z-z_j}{z_j
\gamma_j  \E^{2t\Phi(z_j)} }\\ 0 &1\end{pmatrix}[P(z)]^{-\sigma_3}, & \quad  |z-z_j|< \delta, \quad z_j\in(-1,0),\\
\begin{pmatrix} 1 & 0 \\
\frac{z_j \gamma_j \E^{2t\Phi(z_j)} }{z-z_j} & 1\end{pmatrix}[P(z)]^{-\sigma_3}, & \quad |z-z_j|< \delta, \quad z_j\in(0,1),\\
\si_1 A(z^{-1})\si_1, & \quad |z^{-1}-z_j|< \delta, \quad j=1, \dots, N, \\
[P(z)]^{-\sigma_3}, & \quad \mbox{else.} \\
\end{cases}
\]
We consider the circles $\mathbb T_j$ as contours with counterclockwise orientation. 
Denote their images under the map $z\mapsto z^{-1}$ by $\mathbb T_j^*$ and orient them clockwise. 
These curves are not circles, but they  surround $z_j^{-1}$ with minimal distance from the curve to $z_j$ given by $\frac{\delta}{z_j(z_j - \delta)}$. 
We redefine the vector $m$ by 
\[
m^{\mathrm{ini}}(z)=m(z) A(z), \quad z\in\C.
\]
Then $m^{\mathrm{ini}}(z)$ is a holomorphic function in $\C\setminus\{\T\cup I\cup I^*\cup \T^\delta\}$,
\[
\T^\delta := \bigcup_{j=1}^N \T_j \cup \T_j^*,
\]
and solves the jump problem 
\[
m_+^{\mathrm{ini}}(z)=m_-^{\mathrm{ini}}(z)B(z), \quad z\in\T^\delta, 
\]
in neighborhoods of the discrete spectrum, where (cf.\ \cite{KTb})
\be\label{defB}
B(z)= \begin{cases} 
\begin{pmatrix}1& \frac{(z-z_j)P^2(z)}{z_j
\gamma_j  \E^{2t\Phi(z_j)} }\\ 0 &1\end{pmatrix}, &  \quad z\in\T_j, \quad z_j\in(-1,0),\\
\begin{pmatrix} 1 & 0 \\
\frac{z_j \gamma_j \E^{2t\Phi(z_j)} }{(z-z_j)P^2(z)} & 1\end{pmatrix}, &  \quad z\in\T_j, \quad z_j\in(0,1), \\
\sigma_1(B(z^{-1}))^{-1}\sigma_1, & \quad z\in\T_j^*, \quad j=1, \dots, N.
\end{cases}
\ee
Note that the matrix $B(z)$ satisfies the symmetry condition \eqref{matsym} and 
\[
\|B(z)-\id\|\leq C\E^{- t |\inf_j\re\Phi(z_j)|}, \quad 
z\in \T^\delta, \quad 0<\xi<1. 
\]
Here the matrix norm is to be understood as the maximum of the absolute value of its elements.

Consider the contour $\Gamma=\T \cup I\cup I^*\cup \T^\delta$, where the unit circle $\T$ is oriented 
counterclockwise and the intervals $I$, $I^*$ are oriented towards the center of the circle.
Continue the function \eqref{defchi} to $I^*$ by $\chi(z)=-\chi(z^{-1})$.
Then the following proposition is valid  (cf.\ \cite{dkkz, KTb}).

\begin{proposition}\label{thm:vecrhp} 
Suppose that the initial data of the Cauchy problem \eqref{tl}--\eqref{main} satisfy \eqref{decay}.
Let $\{ R(z), |z|=1; \chi(z), z \in I; (z_j, \gamma_j), 1 \leq j \leq N\}$ be
the right scattering data of the operator $H(0)$. Suppose that $H(0)$ has no resonances at 
the spectral edges $b-2a$, $b+2a$. Let $\xi=\frac{n}{t}\in (0,1)$. Then the vector-valued function 
$m^{\mathrm{ini}}(z)=m^{\mathrm{ini}}(z,n,t)$, connected with the initial function \eqref{defm} by 
\be\nn%\label{inverse1} 
m^{\mathrm{ini}}(z)=m(z)[P(z)]^{-\sigma_3}, \quad z\in\{z : |z^{\pm 1} - z_j|>\delta, \ 1 \leq j \leq N\},
\ee 
is the unique solution of the following vector Riemann--Hilbert problem:
Find a vector-valued function $m^{\mathrm{ini}}(z)$ which is holomorphic away from $\Gamma$, 
continuous up to the boundary, and satisfies:
\begin{enumerate}[I.]
	\item  The jump condition $m_{+}^{\mathrm{ini}}(z)=m_{-}^{\mathrm{ini}}(z) v(z)$,  where
\be \nn%\label{eq:jumpcond}
v(z)=\left\{
\begin{array}{ll}
\begin{pmatrix}
0 & - P^2(z)\ol{R(z)} \E^{- 2 t \Phi(z)} \\
P^{-2}(z)R(z) \E^{2 t \Phi(z)} & 1
\end{pmatrix}, & \quad z \in \T,\\[3mm]
\begin{pmatrix}
1 & 0 \\
P^{-2}(z)\chi(z) \E^{2t\Phi(z)} & 1
\end{pmatrix}, & \quad z \in I,\\[3mm]
\begin{pmatrix}
1 & P^2(z)\chi(z) \E^{-2t\Phi(z)} \\
0 & 1
\end{pmatrix}, & \quad z \in  I^*,\\[3mm]
B(z),  & \quad z \in \T^\delta.
\end{array}\right.
\ee
The phase function $\Phi(z)=\Phi(z, n/t)$ is given by \eqref{Phi}, the matrix  $B(z)$ by \eqref{defB}, 
the function $P(z)$ by \eqref{Blaschke}, and $\chi(z)$ by \eqref{defchi}.
\item
The symmetry condition $m^{\mathrm{ini}}(z^{-1}) = m^{\mathrm{ini}}(z) \si_1$.
\item
The normalization condition
\be\label{eq:normcond}
m^{\mathrm{ini}}(0) = (m_1^{\mathrm{ini}}, m_2^{\mathrm{ini}}), \quad m_1^{\mathrm{ini}} \cdot m_2^{\mathrm{ini}} = 1, 
\quad m_1^{\mathrm{ini}} > 0.
\ee
\end{enumerate}
\end{proposition}
\begin{remark}\label{rem1} (i). Note that the matrix $v(z)$ satisfies the symmetry property $v(z)=\sigma_1 (v(z^{-1}))^{-1}\sigma_1$ .
Moreover, in Proposition \ref{thm:vecrhp} we assume that the points $b-2a$ and $b+2a$ are non-resonant. This means that the initial vector function $m$ 
has continuous limits on $I$, $I^*$ and that $\chi(z)$ is bounded there (otherwise $T(z)\sim (z-q_j)^{-1/2}$ and both $m$ 
and the jump matrix have singularities). For $\xi\in (0,1)$, we have $I\subset \{z : \re \Phi(z)=\Phi(z)<0\}$.  Therefore
\be\nn%\label{invr}
\|v(z)-\id\|_{L^\infty(I\cup I^*)}\leq C\E^{-t |\Phi(q_1)|}.
\ee 
(ii). The normalization condition $m_1^{\mathrm{ini}}>0$ holds
 since it holds for the initial function $m$ and by definition, $P(0)>0$.
\end{remark}
 
We  omit the proof of Proposition \ref{thm:vecrhp} which is essentially the same as in
 \cite{dkkz, KTb}. Uniqueness of the solution can be proven as in \cite{aelt}.

\section{Reduction to the model problem}
\label{sec:model}
In this section we perform three conjugation-deformation steps which reduce the RH problem I--III 
to a simple jump problem on an arc of the circle with a constant jump matrix, 
plus jump matrices which are small with respect to $t$. The jump problem with the constant matrix can be solved 
explicitly. Note that since $|R(z)|=1$ for $z\in \T$, we cannot apply the standard lower-upper factorization 
of the jump matrices on an arc of $\T$ and the subsequent "lens" machinery near this arc (see \cite{KTb}).
For this reason we first have to find a suitable $g$-function (\cite{dvz}).

\vskip 1mm
{\it Step 1.} In this step we replace the phase function by a function with ``better'' properties. The $g$-function has the same asymptotics (up to a constant term) as $\Phi(z)$ for $z\to 0$ and $z\to\infty$ and the same oddness property, $g(z^{-1})=-g(z)$. In addition, it has the convenient property that the curves separating the domains with 
different signs of $\re g(z)$ cross at $z_0(\xi)\in\T$ and $\ol z_0(\xi)$, where $g(z_0)=g(\ol z_0)=0$. A second 
helpful property is that the $g$-function has a jump along the arc connecting $z_0$ and $\ol z_0$ which satisfies $g_+(z) =- g_-(z)>0$. This simplifies further transformations, because with this property and Lemma \ref{lem:conjug} we do not need the lens machinery around this arc.  Note that $z_0$ does not coincide with the stationary phase point of $\Phi$.
Recall that for $\xi\in (0,1)$, the curves $\re \Phi(z)=0$ intersect at the symmetric points 
$-\xi\pm\sqrt{\xi^2 - 1}\in\T$, which are the stationary points of $\Phi$. That is, the stationary phase point 
corresponds to the angle $\phi_0\in (0, \pi)$ where $\cos\phi_0=-\xi$. Set
\be\label{defz0}z_0=\E^{\I \theta_0}, \ \mbox{where}\ \cos \theta_0=1-2\xi,\quad \theta_0\in(0,\pi),
\ee 
and introduce
\be\label{g1}
g(z) = \frac{1}{2} \int_{z_0}^z \sqrt{\Big(1-\frac{1}{sz_0}\Big)\Big(1 - \frac{z_0}{s}\Big)} (1+s) \frac{ds}{s},
\ee
where $\sqrt{s}>0$ for $s> 0$. The cut of the square root in \eqref{g1} is taken between the points 
$z_0$ and $\ol {z_0}$ along the arc
\be\nn%\label{siggm} 
\Sigma=\{ z\in\T: \re z\leq \re z_0=\cos\theta_0\}.
\ee
We orient $\Sigma$ in the same way as $\T$, i.e., from $z_0$ to $\ol z_0$.

\begin{lemma} \label{lem3.1} The function $g(z)$ defined in \eqref{g1} satisfies the following properties:
	
\begin{enumerate} [{\rm (a)}]
\item $\lim_{z\to\infty}(\Phi(z)-g(z)) = - K(\xi)\in\R$;
\item  $g(\ol{z_0})=0$; 
\item  $g(z^{-1}) = - g(z)$ for $z\in\C\setminus\Sigma$;
\item  It has a jump along the arc $\Sigma$ with $g_+(z)=- g_-(z) > 0$ for $z\neq z_0^{\pm 1}$;
\item In a vicinity of $z_0$, 
\be \nn%\label{gz_0}
g(z)= C(\theta_0) (z-z_0)^{3/2} \E^{\I (\frac{\pi}{4}- \frac{3\theta_0}{2})}(1 + o(1)), \quad C(\theta_0)>0.
\ee
In particular, 
$g(z) =  C(\theta_0) (|z-z_0|)^{3/2}(1+o(1))$ for $z\in\Sigma$.
\end{enumerate}
\end{lemma}

\begin{proof} We first prove that our choice of $z_0$ yields $\Phi(z)-g(z)=O(1)$ as $z\to\infty.$
To match the asymptotics of $g(z)$ and $\Phi(z)$ in \eqref{Phi} we compute
\[
\frac{d}{dz}g(z) = \frac{1}{2} \sqrt{\Big(1-\frac{1}{zz_0}\Big)\Big(1 - \frac{z_0}{z}\Big)} \frac{1+z}{z} 
	=  \frac{1}{2} \Big( 1 - \frac{1}{2 zz_0}- \frac{z_0}{2z} + \frac{1}{z} \Big) + O\Big(\frac{1}{z^2} \Big),
\]
as $z \to \infty$, and hence
\[
g(z) = \frac{z}{2} + \frac{1}{2} \Big( 1 - \frac{z_0 + z_0^{-1}}{2}\Big) \log z + O(1).
\]	
Now choose 
\[
\frac{z_0 + z_0^{-1}}{2} = 1- 2 \xi,
\]
which implies that $\cos \theta_0 = 1-2 \xi$ as desired. To prove property (b),
substitute $s=\E^{\I \theta}$ and $z_0=\E^{\I \theta_0}$ in \eqref{g1}, then
\begin{align*}
g(\ol{z_0}) &= \I \int_{\theta_0}^{2\pi - \theta_0} \sqrt{(1 - \E^{-\I(\theta_0+\theta)})(\E^{\I \theta}- \E^{\I\theta_0})}\ 
\frac{\E^{\frac{\I \theta}{2}} + \E^{- \frac{\I \theta}{2}}}{2} d \theta \\
&= \I \sqrt{2} \int_{\theta_0}^{2\pi - \theta_0} \sqrt{\cos \theta - \cos \theta_0} \cos\tfrac{\theta}{2} d\theta. 
\end{align*}
For $\theta \in [\theta_0, 2\pi - \theta_0]$, we have that $f(\theta)=\cos \theta_0 - \cos \theta>0$ and $f$ is even with respect to $\alpha$, $f(\pi - \alpha)=f(\pi + \alpha)$. Moreover, $\cos(\tfrac{\pi + \alpha}{2})$ is odd with respect to 
$\alpha$, $\cos(\tfrac{\pi - \alpha}{2}) = - \cos(\tfrac{\pi + \alpha}{2})$. Hence the substitution 
$\theta=\pi + \alpha$ with $-\pi + \theta_0 \leq \alpha < \pi - \theta_0$ yields 
\[
g(\ol{z_0}) = \sqrt{2} \int_{\theta_0 - \pi}^{\pi - \theta_0} \sqrt{f(\pi + \alpha)} \cos(\tfrac{\pi + \alpha}{2})d\alpha = 0.
\]
Now it is straightforward to see (c), because the integrand can be represented as
\be\label{impl}
\frac{d}{ds}g(s) =  \frac{1}{s}\sqrt{\frac{s+s^{-1}}{2}-\frac{z_0+z_0^{-1}}{2}}\,\sqrt{\frac{s+s^{-1}}{2}+1} =: \frac{h(s)}{s}.
\ee
Since $h(s^{-1})=h(s)$, one obtains by replacing $s=t^{-1}$
\[
g(z^{-1})= \int_{z_0}^{z^{-1}} \frac{h(s) ds}{s} =  - \int_{\ol {z_0}}^z \frac{h(t) dt}{t} = - \int_{z_0}^z \frac{h(t) dt}{t}
- g(\ol{z_0}) = - g(z).
\]
For property (d), note that due to their equal asymptotic behavior, the signature table for $g$ as $z \to 0$ or 
$z \to \infty$ is the same as the signature table for $\Phi$ (see Fig.~\ref{fig:g}). 
The line $\T \setminus \Sigma$ corresponds to $\re g=0$.
Indeed, if $z_1 = \E^{\I \theta_1}$ and $\re z_1 > \re z_0$, then
\[
g(z_1) = \I \sqrt{2} \int_{\theta_0}^{\theta_1} \sqrt{\cos \theta - \cos \theta_0} \cos(\tfrac{\theta}{2}) d\theta \in \I \R.
\] 
The function $g(z)$ has a jump along the contour $\Sigma$, but the limiting values are real (compare with the proof of (b)). Thus $\re g_+ = g_+ > 0$ (because this limit is taken from the domain where $\re g > 0$). Respectively, 
$\re g_- = g_- < 0$ and $g_+ = - g_-$.

Now we are ready to finish the proof of property (a). Since $g(z)$ satisfies (c) we have $g(1)=0$. By use of \eqref{impl} and \eqref{spec5}, and taking into account that 
\[
\frac{dz}{z}=-\frac{d\la}{\sqrt{\la^2-1}},\quad \Phi(z)= -\int_1^{\la(z)} \frac{x+\xi}{\sqrt{x^2 -1}}dx,
\]
we obtain for $z\to +0$, that is when $\sqrt{\la^2 - 1}>0$ as $\la>1$, the asymptotic behavior
\be \label{impimp}
\Phi(z) - g(z) = K(\xi) + k(\xi)z + O(z^{2}), \quad \mbox{ as $z\to 0$}.
\ee
Here
\be \label{Kxi}
 K(\xi)=\int_1^\infty \frac{\sqrt {(x - 1+2\xi)(x+1)} - x -\xi}{\sqrt{x^2 - 1}}\, dx \in \R, \quad  k(\xi)=1-2\xi +\xi^2.\ee
Differentiating \eqref{Kxi} we also get 
\be\label{differ}
\frac{d}{d\xi} K(\xi) =-\log\xi.
\ee

To prove (e) we decompose \eqref{g1} in a vicinity of $z_0$ and integrate. 
Since 
$\arg (z-z_0) =( -\frac{3 \pi}{2} + \theta_0)(1 +o(1))$
for $z\in\Sigma$ in a small vicinity of $z_0$, then 
$
\arg (z-z_0)^{3/2}+ \frac{\pi}{4} -\frac{3\theta_0}{2}= - 2 \pi (1+o(1)),
$
which implies the second claim of (e). 
\end{proof}
The signature table for $g(z)$ is given in Fig.~\ref{fig:g}.
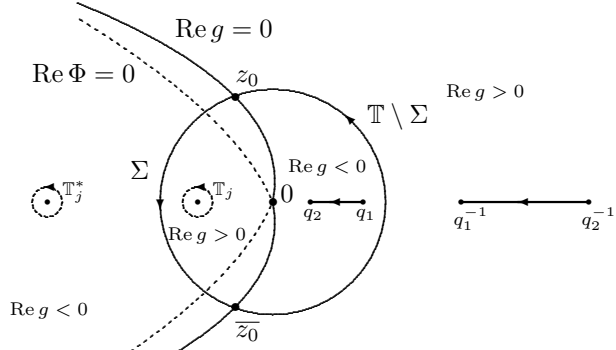
\begin{figure}[ht]
\begin{picture}(6,5.3)
	
	\put(3,2.5){\bigcircle{3}}
	\put(3,2.5){\circle*{0.09}}
	\put(3.1,2.5){$0$}	
	
	\put(2.5,3.9){\circle*{0.09}}
	\put(2.5,4.1){$z_0$}
	\put(2.5,1.1){\circle*{0.09}}
	\put(2.5,0.7){$\ol {z_0}$}

	\put(1.1,2.8){$\Sigma$}
	\put(4.25,3.55){$\T \setminus \Sigma$}
    \put(1.5,2.5){\vector(0,-1){0.1}}
    \put(4.06,3.56){\vector(-1,1){0.1}}

\linethickness{0.3mm}
\put(3.5,2.5){\line(1, 0){0.7}}
\put(3.8,2.5){\vector(-1,0){0.05}}
\put(3.5,2.5){\circle*{0.07}}
\put(3.4,2.3){\scriptsize {$q_2$}}
\put(4.2,2.5){\circle*{0.07}}
\put(4.1,2.3){\scriptsize {$q_1$}}

\put(5.5,2.5){\line(1, 0){1.7}}
\put(6.3,2.5){\vector(-1,0){0.05}}
\put(5.5,2.5){\circle*{0.07}}
\put(5.4,2.2){\scriptsize {$q_1^{-1}$}}
\put(7.2,2.5){\circle*{0.07}}
\put(7.1,2.2){\scriptsize {$q_2^{-1}$}}

	\linethickness{0.2mm}
	\curve(3,2.5, 2.5,3.9, 0.4,5.15)
	\curve(3,2.5, 2.5,1.1, 0.4,-0.15)
	
%circle around eigenvalues
\curvedashes{0.04,0.04}	
\put(2,2.5){\bigcircle{0.4}}
\put(2,2.5){\circle*{0.07}}
\put(2.2,2.6){\scriptsize $\T_j$}
\put(2,2.7){\vector(-1,0){0.05}}	

\put(0,2.5){\bigcircle{0.4}}
\put(0,2.5){\circle*{0.07}}
\put(0.2,2.6){\scriptsize $\T_j^*$}
\put(0,2.7){\vector(-1,0){0.05}}

\curvedashes{0.05,0.05}	
	\curve(3,2.5, 2.1,3.7, 0.4,4.95)
	\curve(3,2.5, 2.1,1.3, 0.4,0.05)

	\put(1.7,4.7){$\re g = 0$}
	\put(-0.2,4.1){$\re \Phi = 0$}  

	\put(5.3,3.9){\scriptsize {$\re g > 0$}}
	\put(3.2,2.9){\scriptsize {$\re g < 0$}}
	\put(1.6,2){\scriptsize {$\re g > 0$}}
	\put(-0.5,1){\scriptsize {$\re g < 0$}}

\end{picture}
\caption{Signature table for $g(z)$.}\label{fig:g}
\end{figure} 
With this description of the $g$-function we introduce the function 
 $d(z) =\E^{t(\Phi(z) - g(z))}$. It satisfies the conditions of Lemma \ref{lem:conjug}. 
 Let $m^{\mathrm{ini}}(z)$ be the solution of the RH problem I--III. Set
\be\nn% \label{m1}
m^{(1)}(z)=m^{\mathrm{ini}}(z)[ d(z)]^{-\si_3}, 
\ee
then this vector solves the jump problem $m_+^{(1)}(z)=m_-^{(1)}(z)v^{(1)}(z)$ with
\[
v^{(1)}(z)=\left\{
\begin{array}{ll}
\begin{pmatrix}
0 & - \ol{\mathcal R(z)} \E^{- 2 t g(z)} \\
\mathcal R(z) \E^{2 t g(z)} & 1
\end{pmatrix}, & \quad z \in \T \setminus \Sigma,\\
\begin{pmatrix}
0 & - \ol{\mathcal R(z)} \\
\mathcal R(z)  & \E^{-2 t g_+(z)}
\end{pmatrix}, & \quad z \in \Sigma,\\
E(z),  & \quad z \in \Xi, 
\end{array}\right.
\]
where the following notations have been introduced:
\begin{align} \nn%\label{not} 
\mathcal R(z) &:=R(z)P^{-2}(z), \qquad \Xi:=I\cup I^*\cup \T^\delta,\\
\label{defB11}
 E(z)&:= \begin{cases} 
\begin{pmatrix}
1 & 0 \\
P^{-2}(z)\chi(z) \E^{2t g (z)} & 1
\end{pmatrix}, & \quad z \in I, \\
\sigma_1(E(z^{-1}))^{-1}\sigma_1, & \quad z\in I^*, \\
[d(z)]^{\si_3}B(z)[d(z)]^{-\si_3}, &  \quad z\in \T^\delta.
\end{cases}
\end{align}
The matrix $B(z)$ was defined in \eqref{defB}.
Note that in the non-resonant case for $z=q_1$ and $z=q_2$,
\be\nn%\label{estE}
\sup_{z\in I\cup I^*}\|E(z) - \id\|=\|E(z) - \id\|_{L^\infty(I\cup I^*)}\leq C\E^{-t|g(q_1)|}.
\ee
To obtain an analogous estimate on $\T^{\delta}$, we have to adjust the value for $\delta$. 
Denote  
\be\nn\label{defJ}
\inf_{j}\inf\{|\Phi(z_j)|, |g(z_j)|\}=J>0,\quad j=1,\dots,N.
\ee 
Choose $\delta>0$ so small that 
\be\nn%\label{estJ} 
\sup_j \sup_{z \in \T_j}|g(z) - g(z_j)|<\frac{J}{4}, \quad \sup_j \sup_{z \in \T_j}|\Phi(z) - \Phi(z_j)|<\frac{J}{4}.
\ee
Then 
\be\nn%\label{estE_1}
\sup_{z\in \T^\delta}\|E(z)-\id\|=\|E(z) - \id\|_{L^\infty(\T^\delta)}\leq C\E^{-\frac{t J}{2}}.
\ee

\vskip 1mm

{\it Step 2}: On $\T \setminus \Sigma$, the jump matrix can be factorized using the standard upper-lower factorization (\cite{dkkz, KTb}). Let $\mathcal C$ be a contour close to the complementary arc $\T\setminus \Sigma$ with endpoints 
$z_0$ and $\ol z_0$ and clockwise orientation. Let $\mathcal C^*$ be its image under the map $z\mapsto z^{-1}$, oriented clockwise as well.
\begin{figure}[ht]
\begin{picture}(6,5)
	
	\put(3,2.5){\bigcircle{3}}
	\put(3,2.5){\circle*{0.09}}
	\put(3.1,2.5){$0$}	
	
	\put(2.5,3.9){\circle*{0.09}}
	\put(2.3,3.6){$z_0$}
	
	\put(2.5,1.1){\circle*{0.09}}
	\put(2.2,1.3){$\ol {z_0}$}

\put(1.05,2.4){$\Sigma$}
  \put(1.5,2.5){\vector(0,-1){0.1}}
	%\put(2.75,4.1){\scriptsize{$\T \setminus \Sigma$}}

	\linethickness{0.2mm}
	\curve(3,2.5, 2.5,3.9, 1,5)
	\curve(3,2.5, 2.5,1.1, 1,0)
	\put(1.7,4.7){$\re g = 0$}

	%half circle rad 1
	%\curve(3,3.5, 3.71,3.21, 4,2.5, 3.71,1.79, 3,1.5)
	%\put(3,3.5){\circle*{0.09}}
	%\put(3.71,3.21){\circle*{0.09}}
    %\put(3.71,1.79){\circle*{0.09}}
	%\put(3,1.5){\circle*{0.09}}

    \curvedashes{0.04,0.04}
	%half cirlce rad 1.1
	\curve(3,3.6, 3.77,3.27, 4.1,2.5, 3.77,1.73, 3,1.4)
	\curve(3,3.6, 2.75,3.7, 2.5,3.9)
	\curve(3,1.4, 2.75,1.3, 2.5,1.1)
	\put(3.77,3.27){\vector(1,-1){0.1}}
	%\put(4.1,2.5){\vector(0,-1){0.1}}

	%half cirlce rad 1.9
	\curve(3,4.4, 4.34,3.84, 4.9,2.5, 4.34,1.16, 3,0.6)
	\curve(3,4.4, 2.55,4.25, 2.5,3.9)
	\curve(3,0.6, 2.55,0.75, 2.5,1.1)
	\put(4.34,3.84){\vector(1,-1){0.1}}
	%\put(4.9,2.5){\vector(0,-1){0.1}}

    \put(4.15,2.1){$\Omega$}
    \put(3.42,3){$\mathcal C$}
    \put(4.4,1.7){$\Omega^*$}
    \put(4.55,3.8){$\mathcal C^*$}

	\linethickness{0.3mm}
	\put(3.5,2.5){\line(1, 0){0.4}}
	\put(3.65,2.5){\vector(-1,0){0.05}}
	\put(3.5,2.5){\circle*{0.07}}
	\put(3.9,2.5){\circle*{0.07}}
	\put(5.5,2.5){\line(1, 0){1.7}}
	\put(6.3,2.5){\vector(-1,0){0.05}}
	\put(5.5,2.5){\circle*{0.07}}
	\put(5.4,2.2){\scriptsize {$q_1^{-1}$}}
	\put(7.2,2.5){\circle*{0.07}}
	\put(7.1,2.2){\scriptsize {$q_2^{-1}$}}

\end{picture}
\caption{Contour deformation of Step 2.}\label{fig:contourC}
\end{figure}
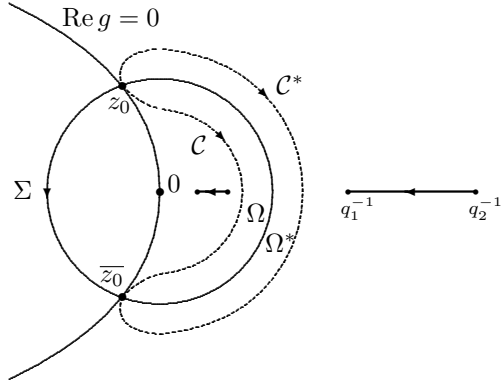
We denote the regions adjacent to these contours by $\Omega$ and $\Omega^*$ as in Fig.~\ref{fig:contourC}, and set
\be \nn%\label{defQ} 
W(z):= \begin{cases}
	\begin{pmatrix}
	1 & 0 \\
	- \mathcal R(z) \E^{2 t g(z)} & 1
	\end{pmatrix}, & \quad z \in \Omega,\\  
	\sigma_1W(z^{-1})\sigma_1, &  \quad z\in\Omega^*.
	\end{cases}
	\ee
Redefine $m^{(1)}$ inside $\Omega$ and $\Omega^*$ by 
$m^{(2)}(z)=m^{(1)}(z) W(z).$  Then the new vector does not have a jump along $\T\setminus \Sigma$, 
and satisfies $m_+^{(2)}(z)=m_-^{(2)}(z)v^{(2)}(z)$ with
\be\nn%\label{v2}
v^{(2)}(z)=\begin{cases} v^{(1)}(z), & \quad z \in \Sigma \cup \Xi,\\
W(z), & \quad z\in \mathcal C,\\
	\sigma_1 (W(z^{-1}))^{-1}\sigma_1, &  \quad z\in \mathcal C^*.
\end{cases}
\ee
The symmetry and normalization conditions are preserved in this deformation step.
\vskip 2mm
Before we perform the next conjugation step, we have to study in detail the solution of the following scalar conjugation problem: {\it find a holomorphic function $\tilde d(z)$ on $\C \setminus \Sigma$, such that }
\begin{align} \label{scp}
& \tilde d_+(z) \tilde d_-(z) = \mathcal R(z) \mathcal R^{-1}(-1), \quad z \in \Sigma,  \\
& \mbox{(i)}\ \tilde d(z^{-1}) = \tilde d^{-1}(z), \quad z \in \C \setminus \Sigma; \quad \mbox{(ii)}\ \tilde d(\infty)>0.\label{symmd}
\end{align}

\begin{remark} 
As for any multiplicative scalar jump problem with non-vanishing jump function on a contour  in  $\C$, one can find its solution  via the associated additive jump problem and the usual Cauchy integral. However, the representation via the usual Cauchy integral requires a prescribed (and known) behavior of the solutions at $\infty$. On the other hand, this representation cannot provide \eqref{symmd}, unless the jump functions are not even on $\Sigma$. For example, condition \eqref{symmd}, (i), implies $\log \tilde d(1)=0$, which cannot be obtained by the usual Cauchy integral. 
That is why we use the Cauchy integral with kernel vanishing at $z=1$,
\be\label{ome}
\Omega(z,s)=\frac{1}{2}\left(\frac{s+z}{s-z} - \frac{s+1}{s-1} \right).
\ee
\end{remark}

In order to solve the conjugation problem \eqref{scp}--\eqref{symmd}, we first solve an auxiliary conjugation problem: find a holomorphic function $F(z)$ in $\C \setminus \Sigma$,  bounded as $z\to\infty$ and such that 
\be\label{aux}F_+(z) = -F_-(z), \ 
z \in \Sigma; \quad F(z^{-1})=-F(z),\ z \in \C \setminus \Sigma;\quad F(\infty)\in\R.
\ee
Since there are two associated additive jump problems,
$
\log F_+(z) = \log F_-(z) \pm \I \pi$, $z \in \Sigma$,
and since
\[
 \frac{\I \pi}{2 \pi \I} \int_{z_0}^{\ol{z_0}} \Omega(z,s)\frac{ds}{s}= 
\log \left(\frac{z_0 z - 1}{z_0-z} \right)^{1/2},
\]
the solution of the jump problem in \eqref{aux} can be given by 
\be \nn%\label{F(z)}
F(z)= \frac{1}{2 \I} \left(\left(\frac{z_0 z - 1}{z_0-z} \right)^{1/2} - \left(\frac{z_0-z}{z_0 z - 1}\right)^{1/2}\right).  
\ee
The symmetry $F(z^{-1})=-F(z)$ is evident from here. It turns out that this symmetry implies that $F_+(s)$ 
is an even function for $s\in\Sigma$. Moreover, 
$F_+(s)$ is real-valued as $s=\E^{\I\theta}$ and $\theta\neq\theta_0$,
\be \label{bes}
F_+(\E^{\I \theta}) = \frac{1}{2 \I} \Bigg(\left(\frac{\sin\tfrac{\theta_0 + \theta}{2}}
{\sin\tfrac{\theta_0-\theta}{2}} \right)^{1/2} - \left(\frac{\sin\tfrac{\theta_0 - \theta}{2}}
{\sin\tfrac{\theta_0+\theta}{2}} \right)^{1/2}\Bigg)=\sqrt{2}\frac{\sin\frac{\theta}{2}\cos\frac{\theta_0}{2}}{\sqrt{\cos\theta_0 - \cos\theta}}\in\R,
\ee 
because  $\cos\theta<\cos \theta_0$ on $\Sigma$.  
On the other hand,
\be\label{Fdef1} 
F(z)=\frac{1}{2\I}\frac{(z_0 + 1)(z-1)}{\sqrt{(z_0 - z)(z_0 z - 1)}}
=\cos\tfrac{\theta_0}{2} \frac{(z-1)}{\sqrt{(z-z_0)(z - \ol z_0 )}},
\ee
and in particular, $F(1)=0$ and $F(\infty)=-F(0)=\cos\frac{\theta_0}{2}\in \R$.

Now assume that $\tilde d(z)$ solves \eqref{scp}--\eqref{symmd} and introduce the function $f(z):=F(z) \log \tilde d(z)$.
Then \eqref{scp}, \eqref{symmd}, and \eqref{aux} imply that this function should solve the jump problem 
\be\label{solf}
f_+(z)=
 f_-(z) + h(z), \quad  z \in \Sigma, \quad  \mbox{where}\  h(z):=F_+(z) \log \tfrac{\mathcal R(z)}{\mathcal R(-1)}.
\ee 
Moreover, $f(z^{-1})=f(z)$ should hold for $z\in\C\setminus\Sigma$ and $f(\infty)\in \R$. The jump \eqref{solf} can be satisfied by
\be\label{deffi}
f(z) = \frac{1}{4 \pi \I}\int_{z_0}^{\ol{z_0}} h(s) \left(\frac{s+z}{s-z} - \frac{s+1}{s-1} \right)  \frac{ds}{s}.
\ee
One has to check that the other two conditions are satisfied too. The function  $\log \frac{\mathcal R(z)}{\mathcal R(-1)}$ is an odd function. Indeed, recall that the reflection coefficient $R(z)$ is a continuous function which satisfies $R(z^{-1})=R^{-1}(z)$, $|R(z)|=1$, and so does $P^2(z)$, which implies  $\mathcal R(z^{-1})=\mathcal 
R^{-1}(z)$. Since $F_+(z)$ is an even function, $h(z)$ is odd on $\Sigma$, and the required evenness of 
$f$ can be directly verified from \eqref{deffi}. 

To check that $f(\infty)\in \R$ recall that $|\mathcal R(z)|=1$ yields $h(z)\in\I \R$. When $z\to\infty$ the first integrand in \eqref{deffi}  vanishes due to oddness of $h$.
For the second we obtain
\[
f(\infty)= - \frac{1}{4 \pi \I}\int_{z_0}^{\ol{z_0}} h(s) \frac{s+1}{s-1} \frac{ds}{s} = 
- \frac{1}{4 \pi \I}\int_{\theta_0}^{2 \pi - \theta_0} h(\theta) \cot(\tfrac{\theta}{2}) d\theta \in \R. 
\]
Since $1\notin\Sigma$, the function $f(z)$ is holomorphic in a vicinity of $z=1$ and has a zero at least of first order there. Hence $F^{-1}(z)f(z)$ is well defined, it satisfies the required additive jump problem for $\log\tilde d(z)$, the symmetry $F^{-1}(z)f(z)=-F^{-1}(z^{-1})f(z^{-1})$, and $F^{-1}(\infty)f(\infty)\in\R$. Thus, $\tilde d(z) = \E^{F^{-1}(z)f(z)}$ solves the conjugation problem \eqref{scp}--\eqref{symmd} and combining 
\eqref{ome}, \eqref{bes}--\eqref{deffi} we obtain
\be\label{12}\tilde d(z)= \exp \left(
\frac{\I\sqrt{(z-z_0 )(z-\ol z_0)}}{\sqrt 2 \pi(z-1)}
\int_{\theta_0}^{2\pi-\theta_0} \arg\frac{\mathcal R(\E^{\I\theta})}{ \mathcal R(-1)}\,\frac{\Omega(z,\E^{\I\theta})\sin\frac{\theta}{2}}{\sqrt{\cos\theta_0 -\cos\theta}} d\theta\right).
\ee
Moreover, taking into account that 
\[
\frac{\Omega(z,s)}{z-1}=\frac{2s}{(s-z)(s-1)},
\] 
and using the first equality in \eqref{Fdef1} we obtain another representation for $\tilde d(z)$,
\be\label{13}
\tilde d(z)=\exp \Bigg( \frac{\sqrt{(z_0 -z)(z_0z-1)}}{2 \pi\I} 
\int_{z_0}^{\ol{z_0}} \frac{ \log\frac{ \mathcal R(s)}{\mathcal R(-1)}\,ds}{[\sqrt{(z_0 -s)(z_0s-1)}]_+(s-z)}\Bigg).
\ee

\begin{lemma}
The solution of the conjugation problem \eqref{scp}--\eqref{symmd} given by \eqref{12} has the following asymptotic behavior as $z\to \infty$
\be\label{asymp}\tilde d(z)=\E^{A(\xi) +\frac{B(\xi)}{z} + O(z^{-2})},
\ee
where
\begin{align}\label{athet}
A(\xi)&=\frac{1}{\sqrt 2\pi} \int_{\theta_0}^{2\pi-\theta_0}\sqrt{\cos\theta_0 -\cos\theta}\frac{2u^\prime(\theta)\sin\frac{\theta}{2} -u(\theta)\cos \frac{\theta}{2}}{4\sin^2\frac{\theta}{2}}d\theta, \\
\label{bthet} B(\xi)& =2 \xi A(\xi) - \frac{1}{\sqrt 2\pi} \int_{\theta_0}^{2\pi-\theta_0}\sqrt{\cos\theta_0 -\cos\theta}\left(2 u^\prime(\theta)
\sin\tfrac{\theta}{2} + u(\theta)\cos\tfrac{\theta}{2}\right)d\theta,
\end{align}
and $u(\theta):=\arg\left(\frac{\mathcal R(\E^{\I\theta})}{\mathcal R(-1)}\right)$.
\end{lemma}

\begin{remark}
The function $\mathcal R$, which is in fact a function of the spectral parameter $\la$, depends on $\theta$ via $\mathcal R(\cos\theta)$, therefore the derivative of $u(\theta)$ is a real valued function.
\end{remark}

\begin{proof}
The function $u(\theta)$ is an odd function of $\theta$ on the interval of integration in the sense that  $u(\frac{\pi}{2}-\alpha)=-u(\frac{\pi}{2}+\alpha)$. In the same sense, 
$\sin\frac{\theta}{2} (\cos\theta_0 -\cos\theta)^{-1/2}$ is an even function. 
Taking this into account as well as \eqref{ome}, we have for $z\to\infty$
\begin{align*}
\I\frac{u(\theta)\Omega(z,\E^{\I\theta})\sin\frac{\theta}{2}}{\sqrt{\cos\theta_0 -\cos\theta}}&=-\I
\left(\frac{\E^{\I\theta}}{\E^{\I\theta}-1}+\frac{\E^{\I\theta}}{z}\right)\frac{u(\theta)\sin\frac{\theta}{2}}{\sqrt{\cos\theta_0 -\cos\theta}} +O(z^{-2})\\
&=\bigg(- \frac{ u(\theta)}{2\sin\frac{\theta}{2}} + \frac{2 u(\theta)\sin\frac{\theta}{2}}{z}\bigg)\frac{d}{d\theta}\sqrt{\cos\theta_0 -\cos\theta} \\
&\quad - \I \frac{u(\theta)\sin\frac{\theta}{2}}{2\sqrt{\cos\theta_0 -\cos\theta}} +\frac{\tilde u(\theta)}{z}
+O(z^{-2}),
\end{align*}
where 
\[
\tilde u(\theta)=\frac{-\I u(\theta)\cos\theta\sin\frac{\theta}{2}}{\sqrt{\cos\theta_0 -\cos\theta}}
\]
is an odd function, $\tilde u(\frac{\pi}{2}-\alpha)=-\tilde u(\frac{\pi}{2}+\alpha)$.
Integration by parts then yields
\be\label{as7}\I\int_{\theta_0}^{2\pi-\theta_0}\frac{u(\theta)\Omega(z,\E^{\I\theta})\sin\frac{\theta}{2}}{\sqrt{\cos\theta_0 -\cos\theta}}d\theta=
\int_{\theta_0}^{2\pi-\theta_0}\sqrt{\cos\theta_0 -\cos\theta}\frac{d}{d\theta}\bigg(
\frac{ u(\theta)}{2\sin\frac{\theta}{2}} - \frac{2 u(\theta)\sin\frac{\theta}{2}}{z}\bigg)d\theta 
\ee 
plus the term of order $z^{-2}$.
On the other hand,
\[
\frac{\sqrt{(z-z_0 )(z-\ol z_0)}}{\sqrt 2 \pi(z-1)}=\frac{1}{\sqrt 2 \pi }\bigg(1 +\frac{1-\cos\theta_0}{z}\bigg) + O(z^{-2}), \quad z\to\infty.
\]
Combining this with \eqref{as7} yields \eqref{athet} and \eqref{bthet}.
\end{proof}

\begin{lemma}\label{lem3.4}
The function $\tilde d(z)$ satisfying \eqref{scp}--\eqref{symmd} has the following asymptotic behavior 
in a vicinity of $z_0$, 
\be \label{dlim}
\tilde d^{-2}(z) \mathcal R(z)= \mathcal R(-1) +O(\sqrt{z-z_0}), \ z \notin \Sigma,
\quad \frac{ \tilde d_+(z)}{ \tilde d_-(z)} = 1+O(\sqrt{z-z_0}), \ z \in \Sigma.
\ee	
\end{lemma} 

\begin{proof}
We will use the representation \eqref{13}. To simplify notation set
 \be\nn%\label{temp}
r(s) = \log( \mathcal R(s)\mathcal R^{-1}(-1)),\quad q(s,z_0)=\sqrt{(z_0-z)(z_0z - 1)}.
\ee 
Then 
\[
\log \tilde d(z)= \frac{q(z,z_0)}{2\pi\I} \int_{z_0}^{\ol z_0} \frac{r(s)ds}{[q(s, z_0)]_+ (s-z)}=J_1(z) + J_2(z),
\]
where
\begin{align*} 
J_1(z)&= \frac{q(z,z_0)}{2\pi\I} \int_{z_0}^{\ol z_0} \frac{(r(s)-r(z))ds}{[q(s, z_0)]_+ (s-z)}, \\ 
J_2(z)&=r(z)\frac{q(z,z_0)}{2\pi\I}\int_{z_0}^{\ol z_0} \frac{ds}{[q(s, z_0)]_+ (s-z)}.
\end{align*}
Since $r(s)-r(z)\sim (s-z)$, the integral in $J_1(z)$ is H\"older continuous in a vicinity
of $z_0$. Therefore,
\be
\begin{aligned} \label{asymp9}
J_1(z) & =I(z_0)\sqrt{z-z_0} (1 +o(1)), \\  
I(z_0)&=\frac{\sqrt{1-z_0^2}}{2\pi\I} \int_{z_0}^{\ol z_0} \frac{(\log\mathcal R(s)-\log\mathcal R(z_0))\,ds}{[\sqrt{(z_0 -s)(z_0s-1)}]_+ (s-z_0)}.
\end{aligned}
\ee
On the other hand, since 
\[
\frac{1}{2[q(z,z_0)]_+}=\frac{1}{2[q(z,z_0)]_-} + \frac{1}{[q(z,z_0)]_+},\quad z\in\Sigma,
\] 
and $(q(z,z_0))^{-1}\to 0$ as $z\to\infty$, we have
\[
\frac{1}{2q(z,z_0)}=\frac{1}{2\pi\I}\int_{z_0}^{\ol z_0} \frac{ds}{[q(s, z_0)]_+ (s-z)},
\]
and $J_2(z)=\frac{r(z_0)}{2}(1 +O(z-z_0))$. Therefore,
\be\nn%\label{asymp17}
\log\tilde d(z)=\frac{1}{2}\log\frac{\mathcal R(z_0)}{\mathcal R(-1)} + I(z_0)\sqrt{z-z_0} + o(\sqrt{z-z_0}),
\ee
and \eqref{dlim} follows from \eqref{asymp9} in a straightforward manner.
\end{proof}

{\it Step 3}: Define $m^{(3)}(z)=m^{(2)}(z)[\tilde d(z)]^{-\sigma_3}$, then our previous considerations lead to the following statement.

\begin{theorem}
The vector function  $m^{(3)}(z)$ is the unique solution of the following RH problem: find a holomorphic 
vector function $\breve m(z)$ in $\C\setminus(\Sigma\cup \mathcal C\cup\mathcal C^*\cup \Xi)$, which is continuous up to 
the boundary and has the following properties:
\begin{itemize} \item It solves the jump problem $\breve m_+(z)=\breve m_-(z)\breve v(z)$ with
\be \label{v3}
\breve v(z)=\left\{\begin{array}{ll}
\begin{pmatrix}
0 & - \mathcal R(-1) \\
\mathcal R(-1) & \frac{\tilde d_+(z)}{\tilde d_-(z)} \E^{-2tg_+(z)}
\end{pmatrix}, & \quad z \in \Sigma,\\
\begin{pmatrix}
1 & 0 \\
-\tilde d^{-2}(z) \mathcal R(z) \E^{2 t g(z)} & 1
\end{pmatrix}, & \quad z \in \mathcal C,\\
\begin{pmatrix}
1 & \tilde d^2(z) \ol{\mathcal R(z)} \E^{- 2 t g(z)} \\
0 & 1
\end{pmatrix}, & \quad z \in \mathcal C^*,\\
\breve E(z):=[\tilde d(z)]^{\sigma_3}E(z)[\tilde d(z)]^{-\sigma_3},  & z \in \Xi. 
\end{array}\right.
\ee
Here $E(z)$ is given by \eqref{defB11}, \eqref{defB}, $\tilde d(z)$ by \eqref{12}, and $\breve E(z)$ satisfies 
\be\nn%\label{theorm1}
\|\breve E(z) - \id\|_{L^\infty(\Xi)}\leq C\E^{-t \breve J},\quad 
\breve J=\min\Big\{|g(q_1)|,\ \frac{J}{2}\Big\},
\ee
where $J$ is defined by \eqref{defJ};
\item $\breve m(z)$ satisfies symmetry and normalization conditions as \eqref{symto} and \eqref{eq:normcond}; moreover,  $\breve v(z)$ has the symmetry property \eqref{matsym};
\item
For small $z$, the vector function $m$ in \eqref{defm} and $\breve m(z)$ are connected by  
\be \label{conn}
\breve m(z)=m(z) [q(z)]^{-\sigma_3},\quad q(z)=\tilde d(z) P(z)
\E^{t(\Phi(z) - g(z))}.
\ee
\end{itemize}\end{theorem}

\section{The model problem solution}

Denote by $\breve\Sigma=\Sigma\cup\mathcal C\cup\mathcal C^*\cup \Xi$ the jump contour for $\breve m$,
and let $\mathcal O$ and $\mathcal O^*$ be small and symmetric (with respect to the map $z\mapsto z^{-1}$) vicinities of $z_0$ and $\ol z_0$, which are not necessarily circles. 
Set $\Sigma_{\mathcal O}=\breve\Sigma\cap \mathcal O$. Let $\Sigma_{\mathcal O}^*$ be the symmetric contour in a neighborhood  of $\ol z_0$. Both contours $\Sigma_{\mathcal O}$ and $\Sigma_{\mathcal O}^*$ inherit the orientation of the respective parts of $\breve\Sigma$.  Evidently,
\be\nn%\label{maintip}
\|\breve v(z)-v^{\mathrm{mod}}(z)\|_{L^\infty\left(\breve\Sigma\setminus(\Sigma_{\mathcal O}\cup\Sigma_{\mathcal O}^*)\right)}\leq C\E^{-t U},
\ee 
where 
\be\nn%\label{est}
U=\min \Big\{\breve J, \inf_{z\in\partial \mathcal O} |g(z)|\Big\}>0,
\ee
and
\be\nn%\label{vmod}
v^{\mathrm{mod}}(z)=\begin{cases}\begin{pmatrix}
0 & -\mathcal R(-1) \\
\mathcal R(-1) & 0
\end{pmatrix}, & z \in \Sigma\\
\id, & z\in\breve\Sigma\setminus\Sigma.\end{cases}
\ee
Thus, in a first order of approximation one can assume that the solution of the  RHP \eqref{v3}
can be approximated by the solution of the following model RHP: {\it find a holomorphic vector function in 
$\C \setminus \Sigma$ satisfying the jump condition
\be \label{modRHP}
m_+^{\mathrm{mod}}(z)= m_-^{\mathrm{mod}}(z) v^{\mathrm{mod}}(z), \quad z \in \Sigma,
\ee
the symmetry condition $m^{\mathrm{mod}}(z^{-1})=m^{\mathrm{mod}}(z)\si_1$, and the
normalization condition} $m_1^{\mathrm{mod}}(0)>0$,   $m_1^{\mathrm{mod}}(0)m_2^{\mathrm{mod}}(0) =1$.

\begin{lemma}
The solution of this vector RH problem is unique.
\end{lemma}
The proof of this Lemma is analogous to the uniqueness proof of the vector model problem 
in \cite{aelt}.

%%%%%%%%%%%%%%%%%%%%%%%%%%%%%%%%%%%%%%%%%%%%%%%%%%%%%%%%%%%%%%%%%%%%%%%%%%%%%%%%%%%%%%%
%%%%%%%%%%%%%%%%%%%%%%%%%%%%%%%%%%%%%%%%%%%%%%%%%%%%%%%%5

For our further investigation we will also need a matrix solution $M^{\mathrm{mod}}(z)=M^{\mathrm{mod}}(z,n,t)$ of the matrix RHP: {\it
find a holomorphic matrix function $M^{\mathrm{mod}}(z)$ on $\C \setminus \Sigma$ satisfying the following jump and symmetry conditions,}
\be\nn%\label{matrixmod}
M_+^{\mathrm{mod}}(z)= M_-^{\mathrm{mod}}(z)
 v^{\mathrm{mod}}(z), \quad z \in \Sigma; \qquad M^{\mathrm{mod}}(z^{-1})=\sigma_1 M^{\mathrm{mod}}(z)\sigma_1.
\ee
 We find a solution of the matrix problem following \cite{its, aelt}. Consider the non-resonant case, that is, 
 $\mathcal R(-1)=-1$. Using 
\be\nn% \label{split1}
v^{\mathrm{mod}}(z)=\begin{pmatrix} 0 & 1 \\ -1 & 0 \end{pmatrix} 
	= \begin{pmatrix} 1 & 1 \\ \I & - \I \end{pmatrix}
	\I \sigma_3
		\begin{pmatrix} 1 & 1 \\ \I & - \I \end{pmatrix}^{-1},
\ee
we first look for a holomorphic solution of the jump problem $M_+^0=\I M_-^0\si_3$
satisfying the symmetry condition as above. Using \eqref{ome} we get
\[
M^0(z)= \begin{pmatrix} \beta(z) & 0 \\ 0 & \beta^{-1}(z) \end{pmatrix}, 
\]
with
\be\nn%\label{betanon}
 \beta(z)=\left(\frac{z_0 z - 1}{z_0 - z}\right)^{1/4}, 
\ee
where the branch of the fourth root is chosen with a cut along the negative half axis and $1^{1/4}=1$. 
Since $\beta(z^{-1})=\beta^{-1}(z),$ we have the required symmetry.
Concerning the original matrix solution $M^{\mathrm{mod}}(z)$, this begs for the representation
\be \label{Mmod}
M^{\mathrm{mod}}(z)=\begin{pmatrix} \frac{\beta(z) + \beta^{-1}(z)}{2} & \frac{\beta(z) - \beta^{-1}(z)}{2 \I} \\[2mm] 
- \frac{\beta(z) - \beta^{-1}(z)}{2 \I} & \frac{\beta(z) + \beta^{-1}(z)}{2} \end{pmatrix},
\ee
and the required symmetry condition is also fulfilled. The vector solution of the model problem is unique. 
Evidently, if we take
 $m^{\mathrm{mod}}(z)= (\alpha, \alpha)M^{\mathrm{mod}}(z)$ for some $\alpha$, 
then \eqref{modRHP} and  \eqref{symto} are fulfilled. We have to choose a suitable $\alpha$ to 
satisfy the normalization condition. Since $\beta(\infty)=\E^{\I\frac{ \theta_0-\pi }{4}}$, then
\be\nn%\label{normm1}
M^{\mathrm{mod}}(\infty)=\begin{pmatrix}\cos\frac{ \theta_0-\pi }{4} & \sin\frac{ \theta_0-\pi }{4}\\ -\sin\frac{ \theta_0-\pi }{4}& \cos\frac{ \theta_0-\pi }{4}\end{pmatrix}.
\ee
The normalization condition for $m^{\mathrm{mod}}(\infty)$ implies 
\[\alpha^2\left(\cos^2\tfrac{\theta_0-\pi }{4} - \sin^2\tfrac{ \theta_0-\pi }{4}\right)\alpha^2\cos \tfrac{\theta_0-\pi}{2}=\alpha^2\sin\tfrac{\theta_0}{2}=1.
\]
Thus,
\be\label{mmod}
m^{\mathrm{mod}}(z)=\alpha \begin{pmatrix} 1, & 1 \end{pmatrix} M^{\mathrm{mod}}(z),
\quad \alpha=\alpha(\xi)=\big(\sin\tfrac{\theta_0}{2}\big)^{-1/2}.
\ee
%%%%%%%%%%%%%%%%%%%%%%%%%%%%%%%%%%%%%%%%%%%%%%%%%%%%%%%
\noprint{
\[m^{\mathrm{mod}}(\infty)=\left(\nu,\,\nu^{-1}\right),\quad \nu=\sqrt{\frac{\cos\frac{ \theta_0-\pi }{4} -\sin\frac{ \theta_0-\pi }{4}}{
 \cos\frac{ \theta_0-\pi }{4} +\sin\frac{ \theta_0-\pi }{4}}}.\]
One can simplify this formula by use of  (\cite{Dwight}, formula 401.06):
\[\nu^2=\frac{\tan\frac{\pi}{4}-\tan\frac{ \theta_0-\pi }{4}}{1+\tan\frac{\pi}{4}\tan\frac{ \theta_0-\pi }{4}}=\tan\left(\frac{\pi}{2}-\frac{\theta_0}{4}\right)=\frac{\sin\frac{\theta_0}{2}}{1 - \frac{\cos\theta_0}{2}}.
 \]
}
%%%%%%%%%%%%%%%%%%%%%%%%%%%%%%%%%%%%%%%%%%%%%%%%%%%%%%%
In the resonant case the matrix solution is represented by the same formula \eqref{Mmod}, but with $\beta^{-1}(z)$ instead of $\beta(z)$ and vice versa. In summary, we have the following

\begin{lemma} \label{lembeta}
The solution of the vector resp.\ matrix model RHP, $m^{\mathrm{mod}}(z)$ resp.\ $M^{\mathrm{mod}}(z)$, is given by  \eqref{mmod} resp.\
\eqref{Mmod}, where 
 $\beta(z)=\big(\frac{z_0 z - 1}{z_0 - z}\big)^{1/4}$\   %$\nu=\left(\frac{\sin\frac{\theta_0}{2}}{1-\cos\frac{\theta_0}{2}}\right)^{1/2}$
 in the non-resonant case and 
 $\beta(z) = \big(\frac{z_0 - z}{z_0 z - 1}\big)^{1/4}$ \
 %$\nu=\left(\frac{\sin\frac{\theta_0}{2}}{1+\cos\frac{\theta_0}{2}}\right)^{1/2}$ 
 in the resonant case.
\end{lemma}
\begin{lemma} The asymptotic behavior of $m_1^{\mathrm{mod}}(z)$ as $z\to 0$ is given by
 \be\label{m1mod} 
 m_1^{\mathrm{mod}}(z)=\frac{\xi^{1/4}}{ (1 + (1 - \xi)^{1/2})^{1/2}}\Big(1 + \big(1 -\xi+ (1 - \xi)^{1/2}\big) z\Big) + O(z^2)\ee
in the non-resonant case. In the resonant case, 
\be\nn%\label{m1modres} 
 m_1^{\mathrm{mod}}(z)=\frac{(1 + (1 - \xi)^{1/2})^{1/2}}{\xi^{1/4}}
\Big(1 + \big(1- \xi -  (1 - \xi)^{1/2}\big) z\Big) + O(z^2).
\ee
\end{lemma}
\begin{proof}
 Consider first the non-resonant case. Since
 \[\beta(z)=\E^{\I\frac{\pi - \theta_0}{4}}\Big(1 - \frac{\I\sin\theta_0}{2} z\Big) + O(z^2),\quad z\to 0,\]
 then
 \[\frac{\beta(z) + \beta^{-1}(z)}{2}=\cos \tfrac{\pi - \theta_0}{4} + \sin\tfrac{\pi - \theta_0}{4}\,\frac{\sin\theta_0}{2}\,z +O(z^2),\]
 \[\frac{\beta(z) - \beta^{-1}(z)}{2\I}=\sin\tfrac{\pi - \theta_0}{4} -\cos\tfrac{\pi - \theta_0}{4}\frac{\sin\theta_0}{2}\,z + O(z^2).\]
By \eqref{mmod},
\begin{align*} m_1^{\mathrm{mod}}(z) &=\frac{1}{\sqrt{\sin\tfrac{\theta_0}{2}}}\Big(\cos \tfrac{\pi - \theta_0}{4}    - \sin\tfrac{\pi - \theta_0}{4}+\big(\cos\tfrac{\pi - \theta_0}{4}+\sin\tfrac{\pi - \theta_0}{4}\big)\frac{\sin\theta_0}{2}\,z \Big) 
 + O(z^2)\\
 & =\frac{\sqrt 2\,\sin\frac{\theta_0}{4}}{\sqrt{\sin\frac{\theta_0}{2}}} +
 \frac{\sqrt 2\,\cos\frac{\theta_0}{4}}{\sqrt{\sin\frac{\theta_0}{2}}}\frac{\sin\theta_0}{2}\,z +O(z^2)=\sqrt{\tan \tfrac{\theta_0}{4}} + \frac{\sin\theta_0\,z}{2\sqrt{\tan \tfrac{\theta_0}{4}}}+O(z^2).
\end{align*}
Moreover,
\[
\tan\tfrac{\theta_0}{4}=\frac{1 -\cos\frac{\theta_0}{2}}{\sin\frac{\theta_0}{2}}=
\frac{\sqrt 2 - \sqrt{1 +\cos\theta_0}}{\sqrt{1 - \cos\theta_0}}=\frac{1 - \sqrt{1- \xi}}{\sqrt{\xi}}=\frac{\sqrt\xi}{ 1 + \sqrt{1 - \xi}},
\]
since $\cos\theta_0=1-2\xi$. Respectively,
$\sin\theta_0=2\sqrt{\xi - \xi^2}$, from which \eqref{m1mod} follows. 
The resonant case follows from the same computation using $\beta^{-1}(z)$ instead of $\beta(z)$.  
\end{proof}

\section{Asymptotics in the region $\xi\in (0,1)$.}
The structure of the matrix solution \eqref{Mmod} and the jump matrix \eqref{v3} as well as the results of 
Lemmas~\ref{lem3.4}, \ref{lem3.1}, (e), allow us to conclude that the solution of the parametrix problem
(which has a local character) can be constructed as in \cite{aelt} using \cite[Appendix B]{KTb} (in particular, the solution can be given in terms of Airy functions).
Consequently, as $z\to 0$ (cf.\ \cite{ep})
\be\nn%\label{asympfin}
m^{(3)}(z,\xi)=\breve m(z)=m^{\mathrm{mod}}(z,\xi) +\frac{F(\xi)}{t}  +z\frac{H(\xi)}{t}+O(z^2) +\frac{G(z,\xi,t)}{t^{4/3}},
\ee
where the vector functions $F(\xi)$ and $H(\xi)$ are H\"older continuous as
$\xi\in\mathcal I:=[\varepsilon, 1-\varepsilon]$ with an exponent $\alpha\geq1/2$.  
The vector-function  $G(z,\xi, t)$ is  uniformly bounded with respect to 
$z\in \overline {\mathbb D_\delta}$,  $\xi\in \mathcal I$ and $t\in[T,\infty)$, 
where $T>0$ is a sufficiently large number and $\mathbb D_\delta$ is a small circle centered at 0. 
By \eqref{conn}, we have
\be\label{des}
m_1(z)=\left(m^{\mathrm{mod}}_1(z) + \frac{F_1(\xi)}{t} + z\frac{H_1(\xi)}{t}\right)\tilde d(z)P(z)\E^{t(\Phi(z)-g(z))}
 + O(t^{-4/3}) +O(z^2).
\ee
Here the term $O(t^{-4/3})$ is uniformly bounded with respect to $z\in \overline {\mathbb D_\delta}$ and $\xi\in \mathcal I$, and the term $O(z^2)$ is uniformly bounded with respect to 
$t\in[T,\infty)$ and $\xi\in \mathcal I$. From \eqref{symmd}, \eqref{asymp}--\eqref{bthet}, \eqref{impimp}, \eqref{Kxi}, 
and \eqref{Blaschke} it follows that
\begin{align} \nn%\label{decomd}
	\tilde d(z)& =\E^{-A(\xi)} \left(1 -B(\xi)z\right) +O(z^2), \quad z\to 0, \\ \nn%\label{decomg} 
	\E^{t(\Phi(z)-g(z))}& =\E^{t K(\xi)}\left(1+t k(\xi) z\right) + O(z^2),\quad z\to 0, \\ \nn%\label{Pe} 
	P(z)&=\bigg(1 - 2z\sum_{z_j \in (-1,0)}\sqrt{\la_j^2 -1}\bigg)\prod_{z_j \in (-1,0)} \frac{1}{|z_j|} + O(z^2).
\end{align} 
Combining this with \eqref{m1mod} and \eqref{des} we get
\be\label{almost}
m_1(z)=\mathcal A(\xi)\left\{1 + 2z\left(\mathcal B(\xi)+\frac{\tilde H_1(\xi)}{t}\right) +\frac{\tilde F_1(\xi)}{t}\right\} + O(t^{-4/3}) +O(z^2),  \ee
where
\begin{align}\nn \mathcal A(\xi,t)&=\E^{tK(\xi) -A(\xi)}\frac{\xi^{1/4}}{ (1 + (1 - \xi)^{1/2})^{1/2}} \prod_{z_j \in (-1,0)}\frac{1}{|z_j|}\\ \nn%\label{aksi}
 &=S\exp \Big(  t K(\xi) - A(\xi) +\frac{1}{4}\log \xi - \frac{1}{2}\log(1 + (1 - \xi)^{1/2})\Big), \\ \nn%\label{coef4} 
\mathcal B(\xi,t) & =-S_1 +\frac{tk(\xi)}{2} - \frac{B(\xi)}{2} + \frac{1 -\xi +(1-\xi)^{1/2}}{2}, \\ \nn
S&=\prod_{z_j \in (-1,0)}\frac{1}{|z_j|},\quad S_1=\sum_{z_j \in (-1,0)}\sqrt{\la_j^2 -1}.
\end{align}
The terms $O(t^{-4/3})$ and $O(z^2)$ are uniformly bounded as above and $\tilde F_1(\xi)$ and $\tilde H_1(\xi)$ are H\"older continuous with an exponent $\alpha\geq 1/2$.
Comparing \eqref{almost} with \eqref{asm1} implies the asymptotics 
\[ 
\prod_{j=n}^\infty 2a(j,t)=\mathcal A(\xi)\Big(1 +\frac{\tilde F_1(\xi)}{t}\Big),
\quad \sum_{m=n}^\infty b(m,t)=\mathcal B(\xi) +\frac{\tilde H_1(\xi)}{t},\quad \xi=\frac{n}{t}.
\]
Set $\hat\xi=\frac{n+1}{t}$ and note that $\xi - \hat\xi=-\frac{1}{t}$.  Then 
\[2a(n,t)=\frac{\mathcal A(\xi)}{\mathcal A(\hat\xi)}\left(1 +\frac{\tilde F_1(\xi) -\tilde F_1(\hat\xi)}{t} \right)+ O(t^{-4/3})=
\E^{  - K^\prime(\xi) +\frac{1}{t}L(\xi) }+ O(t^{-4/3}),
\]
where
\[L(\xi)=\frac{- K^{\prime\prime}(\xi)}{2} + A^\prime(\xi) -\frac{1}{4\xi} - \frac{1}{4(1-\xi  + (1 - \xi)^{1/2})};
\]
the prime denotes the derivative with respect to $\xi$. Analogously,
\[\aligned b(n,t)& =\mathcal B(\xi) -\mathcal B(\hat \xi)+\frac{\tilde H_1(\xi) -\tilde H_1(\hat\xi)}{t} + O(t^{-4/3})\\
& =-\frac{k^\prime(\xi)}{2}+\frac{1}{t}\left(-\frac{k^{\prime\prime}(\xi)}{4} +\frac{B^\prime(\xi)}{2} + \frac12 +\frac{1}{4(1 - \xi)^{1/2}}\right) +O(t^{-4/3}).\endaligned
\]
By \eqref{athet} and \eqref{bthet}
\begin{align}\nn%\label{athet1}
A^\prime(\xi)&=-\frac{1}{4\sqrt 2\,\pi} \int_{\theta_0}^{2\pi-\theta_0}\,\frac{2u^\prime(\theta)\sin\frac{\theta}{2} -u(\theta)\cos \frac{\theta}{2}}{\sqrt{\cos\theta_0 -\cos\theta}\sin^2\frac{\theta}{2}}d\theta, \\ \label{bthet1} B^\prime(\xi)&=2\xi A^\prime(\xi) +2A(\xi) + \frac{1}{\sqrt 2\,\pi} \int_{\theta_0}^{2\pi-\theta_0}\,\frac{2 u^\prime(\theta)
\sin\frac{\theta}{2} + u(\theta)\cos\frac{\theta}{2}}{\sqrt{\cos\theta_0 -\cos\theta}}d\theta,
\end{align}
with (recall that we consider the non-resonant case at $z=-1$)  
\be\label{defu} 
u(\theta)=\arg\left(\mathcal R(\E^{\I\theta}) +1\right).
\ee
Taking into account \eqref{Kxi} and \eqref{differ} we get 
\[ a(n,t)=\frac{\xi}{2}\left( 1 +\frac{L(\xi)}{t} \right) + O(t^{-4/3}),\]
where
\be\nn%\label{lxi}
L(\xi)=-\frac{1 + 2\xi + (1-\xi)^{1/2}}{4\xi(1-\xi + (1 - \xi)^{1/2})} +A^\prime(\xi).\ee
Analogously, by \eqref{Kxi}
\[
b(n,t)=1-\xi +\frac{1}{4t}\big((1-\xi)^{-1/2} + 2 B^\prime(\xi)\big) + O(t^{-4/3}).
\]
In summary, we proved the following 

\begin{theorem}\label{theoras}
Suppose that the initial operator  $H(0)$ associated with the sequences $\{a(n,0), b(n,0)\}$ has no resonance at the edges of the spectrum, $-1, b-2a, b+2a$.
Let $R(z)$, $z=\la - \sqrt{\la^2 -1}$, be its right reflection coefficient and $\la_j=\frac{1}{2}(z_j + z_j^{-1})$ its eigenvalues. 
Let $\varepsilon>0$ be an arbitrary small number. Then in the sector 
\[
\varepsilon t\leq n\leq (1 -\varepsilon t),
\] 
the following asymptotics are valid for the 
solution $ \{a(n,t), b(n,t)\}$ of the Toda lattice as $t\to +\infty$:
\begin{align}\label{amain}
a(n,t)&=\frac{n}{2 t} - \frac{n}{8 t^2} \bigg(\frac{\sqrt{1 -\frac{n}{t}} + 1 +\frac{2n}{t}}
{\frac{n}{t}\left( 1 - \frac{n}{t}+ \sqrt{1 -\frac{n}{t}}\right)}\\ \nn
& \quad +\frac{1}{\sqrt 2 \pi}\int_{\theta_0}^{2\pi - \theta_0} \frac{2u^\prime(\theta)\sin\frac{\theta}{2} -u(\theta)\cos \frac{\theta}{2}}{\sqrt{1-\frac{2n}{t} -\cos\theta}\,\sin^2\frac{\theta}{2}}d\theta\bigg) + O(t^{-4/3}), \\ \label{bmain}
 b(n,t) &= 1-\frac{n}{t} +\frac{1}{4t}\bigg(\frac{1}{\sqrt{1 -\frac{n}{t}}} + 2 B^\prime(\tfrac{n}{t})\bigg) + O(t^{-4/3}),
\end{align}
where the term $O(t^{-4/3})$ is uniformly bounded with respect to $n$. The function $B^\prime(\xi)$ is defined by \eqref{bthet1}, \eqref{bthet}, \eqref{athet}, \eqref{defu};
$\theta_0=\arccos (1 - \frac{2n}{t})\in (0,\pi)$,  and 
\[
u(\theta)=\arg\left(R(\E^{\I\theta})  P^{-2}(\E^{\I\theta}) - \arg R(\E^{\I\pi}) P^{-2}(\E^{\I\pi})\right),
\]
where $P$ is defined in \eqref{Blaschke}.
\end{theorem}

\section{Discussion of the region $-2 a t<n<0$.}

To obtain the asymptotic of the solution in the region  $\frac{n}{t} \in (-2a, 0)$ we could study an analogous 
RH problem connected with the left scattering data and the variable $\zeta$ (cf.\ \eqref{spec5}). Instead of this
extensive analysis let us consider the Toda lattice associated with the functions
\be\label{newab}
\hat a(n, t)= \frac{1}{2a} a\big(-n-1, \tfrac{t}{2a}\big), \quad \hat b(n, t)= 
\frac{1}{2a} \Big(b-b\big(-n, \tfrac{t}{2a}\big)\Big).
\ee
It is straightforward to check that $\hat a(n,t), \hat b(n,t)$ satisfy the Toda equations \eqref{tl} associated with 
the initial profile
\begin{align*} \label{ini2}
\begin{split}	
& \hat a(n,0)\to \frac {1}{2a}, \quad \hat b(n,0) \to \frac{b}{2a}, \quad \mbox{as $n \to -\infty$}, \\
& \hat a(n,0)\to \frac{1}{2}, \quad \hat b(n,0) \to 0, \quad \mbox{as $n \to +\infty$}.
\end{split}
\end{align*}
For this solution we obtain by our previous results in the region $\frac{n}{t} \in (0,1)$ that
\be\label{newas}
\hat a(n,t) = \frac{n}{2t} + \frac {f(n,t)}{t}, \quad \hat b(n,t)=1 - \frac{n}{t}+ \frac{g(n,t)}{t},
\ee
where $f(n,t)$ (resp., $g(n,t)$) are the same as  the second order  terms in \eqref{amain} (resp., \eqref{bmain}), but the function $u$ corresponds to a new Jacobi operator \eqref{ht} with coefficients $\hat H\sim \{\hat a(n,0), \hat b(n,0)\}$. 
Set $\tau = \frac{t}{2a}$ and  $m=-n-1$.
If $\frac{n}{t}\in (0,1)$, then $\frac{m}{\tau}\in (-2a,0)$.
From \eqref{newab} and \eqref{newas} we get 
\begin{align*}
 a(m, \tau) &= 2a \hat a(n,t)= \frac{2a n}{2t} + \frac {2a f(n,t)}{t} =\frac{-m-1}{2\tau} +
 \frac{f(-m-1, 2a\tau)}{\tau}, \\
b(m,\tau) &=b - 2a\hat b(n+1,t)=b- 2a - \frac{m}{\tau} - \frac{g(-m-1, 2a\tau)}{\tau}.
\end{align*}
An elementary analysis shows that the right reflection coefficient of $\hat H$, given in the variable $\zeta$, 
is the same as the left reflection coefficient $R_1(\zeta^{-1})$ of the initial operator $H(0)$, and the discrete spectra of both operators are the same.
In terms of $\zeta$ we denote the eigenvalues by $\zeta_1, \dots,\zeta_N$ and set
\[
P_1(\zeta)= \prod_{\zeta_j \in (0,1)} |\zeta_j| \frac{\zeta - \zeta_j^{-1}}{\zeta-\zeta_j}.
\]
Note that $P(z)$ corresponds to the eigenvalues $\la_j$ of $H(0)$, which satisfy $\la_j<-1$, and $P_1(\zeta)$ corresponds to $\la_j>b+2a$. Now set $\zeta=\E^{\I\theta}$ and
\[U(\theta)=\arg\left( R_1(\E^{\I\theta})P^{-2}_1(\E^{\I\theta})- R_1(1)P^{-2}_1(1)\right),\quad \theta\in [-\theta_0, \theta_0].
\]

\begin{theorem}\label{propos}
Let $\varepsilon>0$ be an arbitrary small number. Suppose that the initial operator $H(0)$ has no resonances at the points $-1,1,b+2a$, and let $R_1(\zeta)$ be its left reflection coefficient, $\la - b =a(\zeta + \zeta^{-1})$. Then in the 
domain $-(2a -\varepsilon)t\leq n\le -\varepsilon t$ the following asymptotic is valid as $t\to +\infty$ in the non-resonant case,
\begin{align}\nn%\label{amainl}
a(n,t)&=-\frac{n+1}{2 t}  - \frac{n+1}{16 a t^2} \left(\frac{\sqrt{1 +\frac{n+1}{2at}} + 1 -\frac{n+1}{at}}
{\frac{n+1}{2at}\left( 1 + \frac{n+1}{2at} + \sqrt{1 +\frac{n+1}{2at}}\right)} \right.\\ \nn
&\qquad \left.-\frac{1}{\sqrt 2 \pi}\int_{-\theta_0}^{ \theta_0} \frac{2U^\prime(\theta)\sin\frac{\theta}{2} -U(\theta)\cos \frac{\theta}{2}}{\sqrt{1+\frac{n+1}{at} -\cos\theta}\sin^2\frac{\theta}{2}}d\theta\right) + O(t^{-4/3}), \\
\nn%\label{bmainl}
 b(n,t) &= b-2a-\frac{n}{t} - \frac{1}{2t}\Bigg(\frac{1}{2 \sqrt{1 + \frac{n+1}{2at}}} 
+ \tilde B(\tfrac nt)\Bigg) + O(t^{-4/3}),
\end{align}
where 
$\theta_0=\arccos (1 + \frac{n+1}{ a t})\in (0,\pi)$ and 
\begin{align*} \tilde B(\tfrac {n}{t}) &= \frac{1}{\sqrt 2 \pi  }\int_{-\theta_0}^{\theta_0}\Bigg(\frac{n+1}{4a t}
 \frac{2U^\prime(\theta)\sin\frac{\theta}{2} -U(\theta)\cos \frac{\theta}{2}}{\sqrt{1+\frac{n+1}{at} -\cos\theta}\sin^2\frac{\theta}{2}}  + \frac{2 U^\prime(\theta)
\sin\frac{\theta}{2}+ U(\theta)\cos\frac{\theta}{2}}{\sqrt{1+\frac{n+1}{at} -\cos\theta}}
\\
&\qquad +
\sqrt{1+\frac{n+1}{at}-\cos\theta}\frac{2U^\prime(\theta)\sin\frac{\theta}{2} -U(\theta)\cos \frac{\theta}{2}}{2\sin^2\frac{\theta}{2}}\Bigg)d\theta.\end{align*}
\end{theorem}

\bigskip
\noindent {\bf Acknowledgment.} I.E.\ is indebted to the Faculty of Mathematics at the University of Vienna for its hospitality and support during the fall semester of 2015, where this work was done.

\end{document}